\documentclass[sigconf]{acmart}

\setcopyright{none}

\usepackage{subfigure} 
\usepackage{enumitem} 
\usepackage{booktabs} 
\usepackage{balance} 
\usepackage{pifont} 
\usepackage[linesnumbered,ruled]{algorithm2e} 
\usepackage{bm} 
\usepackage{breqn} 

\usepackage{colortbl}

\newcommand{\srijan}[1]{{\color{red}[Srijan: #1]}}

\newcommand{\hide}[1]{}
\newcommand{\cut}[1]{}

\newtheorem{conditions}{Condition}

\newcommand{\method}{\textit{JODIE}\xspace}
\newcommand{\batching}{\textit{t-Batch}\xspace}

\begin{document}
\title{Learning Dynamic Embeddings from\\ Temporal Interaction Networks} 

\author{Srijan Kumar}
\affiliation{
	\institution{Stanford University, USA}
}
\email{srijan@cs.stanford.edu}
\author{Xikun Zhang}
\affiliation{
	\institution{University of Illinois, Urbana-Champaign, USA}
}
\email{xikunz2@illinois.edu}
\author{Jure Leskovec}
\affiliation{
	\institution{Stanford University, USA}
}
\email{jure@cs.stanford.edu}


\begin{abstract}
Modeling a sequence of interactions between users and items (e.g., products, posts, or courses) is crucial in domains such as e-commerce, social networking, and education to predict future interactions. Representation learning presents an attractive solution to model the dynamic evolution of user and item properties, where each user/item can be embedded in a euclidean space and its evolution can be modeled by dynamic changes in its embedding. 
However, existing embedding methods either generate static embeddings, treat users and items independently, or are not scalable.
  
Here we propose \method, a coupled recurrent model to jointly learn the dynamic embeddings of users and items from a sequence of user-item interactions.  
\method has three components. 
First, the \textit{update component} updates the user and item embedding from each interaction using their previous embeddings with the two mutually-recursive Recurrent Neural Networks. 
Second, a novel \textit{projection component} is trained to forecast the embedding of users at any future time. 
Finally, the \textit{prediction component} directly predicts the embedding of the item in a future interaction. 
For models that learn from a sequence of interactions, traditional batching of training data can not be done due to complex user-user dependencies. Therefore, we present a novel batching algorithm called \batching that generates time-consistent batches that can be run in parallel, leading to massive speed-up. 

We conduct six experiments to validate \method\ on two prediction tasks---future interaction prediction and state change prediction---using four real-world datasets. We show that \method outperforms six state-of-the-art algorithms in these tasks by up to 22.4\%. Moreover, we show that \method is highly scalable and up to 9.2$\times$ faster than comparable models. As an additional experiment, we illustrate that \method can predict student drop-out from courses up to five interactions in advance. 
\end{abstract}

\maketitle

\vspace{-3mm}
\section{Introduction}

\begin{figure}[t]
\centering
        \includegraphics[width=0.7\columnwidth]{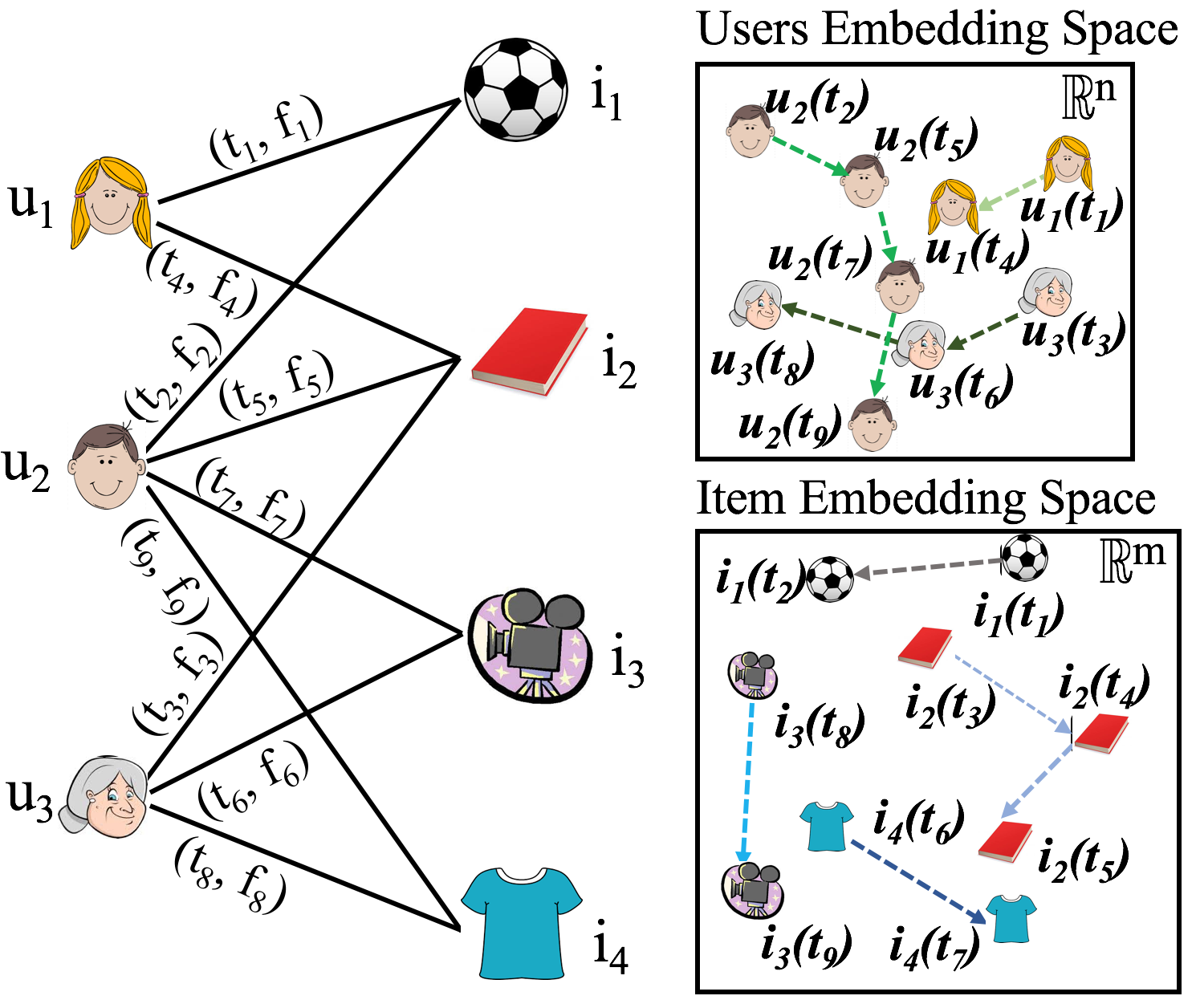}
    \vspace{-3mm}        
    \caption{Left: a network of three users $u_1, u_2$ and $u_3$ interacting with four items $i_1, i_2, i_3$ and $i_4$ over time. Each arrow represents an interaction with associated timestamp $t$ and a feature vector $f$. Right: resulting dynamic embeddings of users and items in the network. } 
    \vspace{-6mm}
    \label{fig:toy}
\end{figure}

Users interact sequentially with items in many domains such as e-commerce (e.g., a customer purchasing an item)~\cite{bobadilla2013recommender,zhang2017deep}, education (a student enrolling in a MOOC course)~\cite{liyanagunawardena2013moocs}, healthcare (a patient exhibiting a disease)~\cite{johnson2016mimic}, social networking (a user posting in a group in Reddit)~\cite{buntain2014identifying}, and collaborative platforms (an editor editing a Wikipedia article)~\cite{iba2010analyzing}. 
The same user may interact with different items over a period of time and these interactions dynamically change over time~\cite{DBLP:journals/debu/HamiltonYL17,DBLP:conf/recsys/PalovicsBKKF14,zhang2017deep,agrawal2014big,DBLP:conf/asunam/ArnouxTL17,raghavan2014modeling,DBLP:journals/corr/abs-1711-10967}. 
These interactions create a dynamic interaction network between users and items. 
Accurate real-time recommendation of items and predicting change in state of users over time are fundamental problems in these domains~\cite{DBLP:conf/wsdm/QiuDMLWT18,DBLP:conf/asunam/ArnouxTL17,DBLP:conf/sdm/LiDLLGZ14,DBLP:journals/corr/abs-1804-01465,DBLP:conf/cosn/SedhainSXKTC13,walker2015complex,DBLP:conf/icwsm/Junuthula0D18}. 
For instance, predicting when a student is likely to drop-out of a MOOC course 	is important to develop early intervention measures for their continued education~\cite{kloft2014predicting,yang2013turn,chaturvedi2014predicting}, and predicting when a user is likely to turn malicious on platforms, like Reddit and Wikipedia, is useful to ensure platform integrity~\cite{kumar2015vews,cheng2017anyone,ferraz2015rsc}. 

\cut{
Let us take an example of a sequence of interactions between users and items using Figure~\ref{fig:toy} (left).
The example has three users $u_1$, $u_2$, and $u_3$ that interact with four items $i_1$, $i_2$, $i_3$ and $i_4$ over time $t_1 - t_9$ (such that $t_i < t_j$ $\forall i < j$). 
Each interaction has an associated feature vector $f_i$, representing the properties of the interaction (e.g., the purchase amount or the number of items purchased). 
We will refer to this example network in the rest of the paper. 
}

Learning embeddings from dynamic user-item interaction networks poses three fundamental challenges. 
We illustrate this using an example interaction network between three users and four items shown in Figure~\ref{fig:toy} (left). 
\textbf{First}, as users interact with items, their properties evolve over time. 
For example, the interest of a user $u_3$ may gradually change from purchasing books (item $i_2$) to movies (item $i_3$) to clothes (item $i_4$). 
Similarly, the properties of items change as different users interact with them. 
For instance, a book (item $i_2$) that is popular in older people (user $u_3$ at time $t_3$) may eventually become popular among the younger audience (users $u_1$ and $u_2$ at times $t_4$ and $t_5$). 
\textbf{Second}, a user's property is influenced by the property of the item that it interacts with and conversely, an item's property is influenced by the interacting user's property. 
For instance, if $u_2$ purchases a book (item $i_3$) after it has won a Pulitzer Prize reflects a different behavior than if $u_2$ purchases $i_3$ before the prize. 
\textbf{Third}, interactions with common items create complex user-to-user dependencies. 
For example, users $u_1$ and $u_2$ interact with item $i_1$, so both users influence each other's properties. 
The traditional methods of batching data during training treat all users independently and therefore they can not be applied to learn embeddings. 
As a result, existing embedding methods have to process the interactions one-at-a-time, and therefore, they are not scalable to a large number of interactions. 
Therefore, modeling the \textit{jointly evolving embeddings of users and items in a scalable way} is crucial in making accurate predictions. 

Representation learning, or learning low-dimensional embeddings of entities, is a powerful approach to represent the dynamic evolution of users and item properties~\cite{DBLP:journals/kbs/GoyalF18,zhang2017deep,dai2016deep,DBLP:conf/nips/FarajtabarWGLZS15,beutel2018latent,zhou2018dynamic}. 
Existing representation learning algorithms, including random-walk based methods~\cite{nguyen2018continuous,perozzi2014deepwalk,grover2016node2vec}, dynamic network embedding methods~\cite{zhou2018dynamic,zhang2017learning}, and recurrent neural network-based algorithm~\cite{wu2017recurrent,beutel2018latent,zhu2017next}, either generate static embeddings from dynamic interactions, learn embeddings of users only, treat users and items independently, or are not scalable to a large number of interactions. 

\textbf{Present work.} 
Here we address the following problem: 
\textbf{given} a sequence of temporal interactions $\mathcal{S}: \mathcal{S}_j = (u_j, i_j, f_j, t_j)$ between users $u_j \in \mathcal{U}$ and items $i_j \in \mathcal{I}$ with a feature vector of the interaction $f_j$ at time $t_j$, 
\textbf{generate} dynamic embeddings \bm{$u_j(t)$} and \bm{$i_j(t)$} for users and items at any time $t$, such that they allow us to solve two prediction tasks: future interaction prediction and user state change prediction. 

\textbf{Present work (\method\ model):}
Here we present an algorithm, \method, which learns dynamic embeddings of users and items from temporal user-item interactions.\footnote{\method\ stands for \underline{Jo}int \underline{D}ynamic User-\underline{I}tem \underline{E}mbeddings.}
Each interaction has an associated timestamp $t$ and a feature vector $f$, representing the properties of the interaction (e.g., the purchase amount or the number of items purchased).  
The resulting user and item embeddings for the example network are illustrated in Figure~\ref{fig:toy} (right). 
We see that \method\ updates the user and item embeddings after every interaction, thus resulting in a dynamic embedding trajectory for each user and item. 
\method\ overcomes the shortcomings of the existing algorithms, as shown in Table~\ref{tab:related}.

In \method, each user and item has two embeddings: a static embedding and a dynamic embedding. 
The static embedding represents the entity's long-term stationary property, while the dynamic embedding represent evolving property and are learned using the \method\ algorithm. 
This enables \method\ to make predictions from both the temporary and stationary properties of the user.

The \method model consists of three major components in its architecture: an update function, a project function, and a predict function. 

The \textit{update function} of \method\ has two Recurrent Neural Networks (RNNs) to generate the dynamic user and item embeddings. 
Crucially, the two RNNs are coupled to explicitly incorporate the interdependency between the users and the items. 
After each interaction, the user RNN updates the user embedding by using the embedding of the interacting item. Similarly, the item RNN uses the user embedding to update the item embedding. 
It should be noted that \method\ is easily extendable to multiple types of entities, by training one RNN for each entity type. 
In this work, we apply \method\ to the case of bipartite interactions between users and items.

A major innovation of \method\ is that it learns a \textit{project function} to forecast the embedding of users at any future time. 
Intuitively, the embedding of a user will change slightly after a short time elapses since its previous interaction (with any item), but the embedding can change significantly after a long time elapses.
As a result, the embedding of the user needs to be estimated for accurate real-time predictions. 
To solve this challenge, \method\ learns a project function that estimates the embedding of a user after some time $\Delta$ elapses since its previous interaction. 
This function makes \method\ truly dynamic as it can generate dynamic user embeddings at any time. 

Finally, the third component of \method\ is the \textit{predict function} that predicts the future interaction of a user. 
An important design choice here is that the function directly outputs the embedding of the item that a user is most likely to interact with, instead of a probability score of interaction between a user and an item. 
As a result, we only need to do the expensive neural network forward pass \textit{once} in our model to generate the predicted item embedding and then find the item that has the embedding closest to the predicted embedding. 
On the other hand, existing models need to do the expensive forward pass $N$ times (once for each candidate item) and select the one with the highest score, which hampers its scalability.

\textbf{Present work (\batching\ algorithm):}
Training models that learn on a sequence of interactions is challenging due to two reasons: (i) interactions with common items results in complex user-to-user dependencies, and (ii) the interactions should be processed in increasing order of their time. 
The naive solution to generate dynamic embeddings is to process each interaction sequentially, which is not scalable to a large number of interactions such as in DeepCoevolve~\cite{dai2016deep} and Zhang et al.~\cite{zhang2017learning}. 
Therefore, we propose a novel batching algorithm, called \batching, that creates batches such that the interactions in each batch can be processed in parallel while still maintaining all user-to-user dependencies. 
Each user and item appears at most once in every batch, and the temporally-sorted interactions of each user (and item) appear in monotonically increasing batches. 
Batching in such a way results in massive parallelization. 
\batching is a general algorithm that is applicable to any model that learns on a sequence of interactions. 
We experimentally validate that \batching leads to a 8.5$\times$ and $7.4\times$ speed-up in the training time of \method\ and DeepCoevolve~\cite{dai2016deep}. 

\textbf{Present work (experiments):}
We conduct six experiments to evaluate the performance of \method\ on two tasks: predicting the next interaction of a user and predicting the change in state of users (when a user will be banned from social platforms and when a student will drop out from a MOOC course). 
We use four datasets from Reddit, Wikipedia, LastFM, and a MOOC course activity for our experiments. 
We compare \method\ with six state-of-the-art algorithms from three categories: recurrent recommender algorithms~\cite{zhu2017next,beutel2018latent,wu2017recurrent}, dynamic node embedding algorithm~\cite{nguyen2018continuous}, and co-evolutionary algorithm~\cite{dai2016deep}. 
\method\ outperforms the best baseline algorithms on the interaction prediction task by up to 22.4\% and up to 4.5\% in predicting user state change. 
We further show that \method outperforms existing algorithms irrespective of the percentage of the training data and the size of the embeddings. 
As an additional experiment, we show that \method\ can predict which student will drop-out of a MOOC course as early as five interactions in advance. 

Overall, in this paper, we make the following contributions:
\vspace{-2mm} 
\begin{itemize}
\item \textbf{Embedding algorithm:} We propose a coupled Recurrent Neural Network model called \method\ to learn dynamic embeddings of users and items from a sequence of temporal interactions. A major contribution of \method\ is that it learns a function to project the user embeddings to any future time. 
\item \textbf{Batching algorithm:} We propose a novel \batching algorithm to create batches that can be run in parallel without losing user-to-user dependencies. This batching technique leads to 8.5$\times$ speed-up in \method\ and 7.4$\times$ speed-up in DeepCoevolve. 
\item \textbf{Effectiveness:} \method\ outperforms six state-of-the-art algorithms in predicting future interactions and user state change predictions, by performing up to 22.4\% better than six state-of-the-art algorithms. 
\end{itemize}
\vspace{-2mm}

\section{Related Work}
\label{sec:related}

Here we discuss the research closest to our problem setting spanning three broad areas. 
Table~\ref{tab:related} compares their differences. 
Any algorithm that learns from sequence of interactions should have the following properties: it should be able to learn dynamic embeddings, for both users and items, in such a way that they are inter-dependent, and the method should be scalable. 
The proposed model \method\ satisfies all the desirable properties. 

\textbf{Deep recommender systems.}
Several recent models employ recurrent neural networks (RNNs) and variants (LSTMs and GRUs) to build recommender systems. 
RRN~\cite{wu2017recurrent} uses RNNs to generate dynamic user and item embeddings from rating networks.  
Recent methods, such as Time-LSTM~\cite{zhu2017next} and LatentCross~\cite{beutel2018latent} incorporate features directly into their model.
However, these methods suffer from two major shortcomings. First, they take one-hot vector of the item as input to update the user embedding.
This only incorporates the item id and ignores the item's current state. 
The second shortcoming is that models such as Time-LSTM and LatentCross generate embeddings only for users, and not for items. 

\method\ overcomes these shortcomings by learning dynamic embeddings for both users and items in a mutually-recursive manner. 
In doing so, \method\ outperforms the best baseline algorithm by up to 22.4\%.

\textbf{Dynamic co-evolution models.} 
Methods that jointly learn representations of users and items have recently been developed using point-process modeling~\cite{wang2016coevolutionary,trivedi2017know} and RNN-based modeling~\cite{dai2016deep}. 
The basic idea behind these models is similar to \method ---user and item embeddings influence each other whenever they interact.
However, the major difference between \method\ and these models is that \method\ learns a projection function to generate the embedding of the entities whenever they are involved in the interaction, while the \textit{projection function} in \method\ enables us to generate an embedding of the user at any time. As a result, we observe that \method\ outperforms DeepCoevolve by up to 57.7\% in both prediction tasks of next interaction prediction and state change prediction. 

In addition, these models are not scalable as traditional methods of data batching during training can not be applied due to complex user-to-user dependencies. 
\method\ overcomes this limitation by developing a novel batching algorithm, \batching, which makes \method\ 9.2$\times$ faster than DeepCoevolve.

\textbf{Temporal network embedding models.}  
Several models have recently been developed that generate embeddings for the nodes (users and items) in temporal networks. 
CTDNE~\cite{nguyen2018continuous} is a state-of-the-art algorithm that generates embeddings using temporally-increasing random walks, but it generates one final static embedding of the nodes, instead of dynamic embeddings. 
Similarly, IGE~\cite{zhang2017learning} generates one final embedding of users and items from interaction graphs. 
Therefore, both these methods (CTDNE and IGE) need to be re-run for every a new edge to create dynamic embeddings. 
Another recent algorithm, DynamicTriad~\cite{zhou2018dynamic} learns dynamic embeddings but does not work on interaction networks as it requires the presence of triads. 
Other recent algorithms such as DDNE~\cite{DBLP:journals/access/LiZYZY18}, DANE~\cite{DBLP:conf/cikm/LiDHTCL17}, DynGem~\cite{goyal2018dyngem}, Zhu et al.~\cite{zhu2016scalable}, and Rahman et al.~\cite{DBLP:journals/corr/abs-1804-05755} learn embeddings from a sequence of graph snapshots, which is not applicable to our setting of continuous interaction data. 
%
Recent models such as NP-GLM model~\cite{DBLP:journals/corr/abs-1710-00818}, DGNN~\cite{DBLP:journals/corr/abs-1810-10627}, and DyRep~\cite{trivedi2018representation} learn embeddings from persistent links between nodes, which do not exist in interaction networks as the edges represent instantaneous interactions.

Our proposed model, \method\, overcomes these shortcomings by generating dynamic user and item embeddings. In doing so, \method\ also learns a \textit{projection function} to predict the user embedding at a future time point. Moreover, for scalability during training, we propose an efficient training data batching algorithm that enables learning from large-scale interaction data. 
%


{
\footnotesize
\begin{table}[t]
\center
\caption{\label{tab:related} Table comparing the desired properties of the existing class of algorithms and our proposed \method\ algorithm. \method\ satisfies all the desirable properties.}
\begin{tabular}{c|c|c|c|c|c|c|c||c}
\hline 
& \multicolumn{4}{c|}{Recurrent} & \multicolumn{3}{c||}{Temporal network} & Proposed \\
& \multicolumn{4}{c|}{models} & \multicolumn{3}{c||}{embedding models} & model \\\cline{2-9}
Property & \rotatebox[origin=c]{90}{LSTM} & \rotatebox[origin=c]{90}{Time-LSTM~\cite{zhu2017next}} & \rotatebox[origin=c]{90}{RRN~\cite{wu2017recurrent}} & \rotatebox[origin=c]{90}{LatentCross~\cite{beutel2018latent}} & \rotatebox[origin=c]{90}{CTDNE~\cite{nguyen2018continuous}} & \rotatebox[origin=c]{90}{IGE~\cite{zhang2017learning}} & \rotatebox[origin=c]{90}{DeepCoevolve~\cite{dai2016deep}} & \rotatebox[origin=c]{90}{\method} \\\hline
Dynamic embeddings & \ding{52} & \ding{52} & \ding{52} & \ding{52} & & & \ding{52} & \ding{52} \\
Embeddings for users and items & & & \ding{52} & & \ding{52}  & \ding{52} & \ding{52} & \ding{52} \\
Learns joint embeddings & & & & & & \ding{52} & \ding{52} & \ding{52} \\
Parallelizable/Scalable & \ding{52}  &  \ding{52} &  \ding{52} & \ding{52}  &\ding{52} & & & \ding{52} \\\hline 
\end{tabular}
\end{table}
}

\section{JODIE: Joint Dynamic User-Item Embedding Model}
\label{sec:method}

In this section, we propose \method, a method to learn dynamic representations of users and items from a sequence of temporal user-item interactions $\mathcal{S}: \mathcal{S}_j = (u_j, i_j, f_j, t_j)$. 
An interaction $S_j$ happens between a user $u_j \in \mathcal{U}$ and an item $i_j \in \mathcal{I}$ at time $t_j$. 
Each interaction has an associated feature vector $f_j$. 
The desired output is to generate dynamic embeddings $\bm{u_j(t)}$ for user $u_j$ and $\bm{i_j(t)}$ for item $i_j$  at any time $t$. 
Table~\ref{tab:symbols} lists the symbols used.


Our proposed model, called \method\, is a dynamic embedding learning method that is reminiscent of the popular Kalman Filtering algorithm~\cite{julier1997new}.\footnote{Kalman filtering is used to accurately measure the state of a system using a combination of system observations and state estimates given by the laws of the system.} 
Like the Kalman filter, \method\ uses the interactions (i.e., observations) to update the state of the interacting entities (users and items) via a trained \textit{update function}. 
A major innovation in \method\ is that between two observations of a user, its state is estimated by a trained \textit{projection function} that uses its previous observed state and the elapsed time to generate a \textit{projected}
When the entity's next interaction is observed, its new states are updated again. 

The \method\ model is trained to accurately predict future interactions between users and items. 
Instead of predicting a probability score of interaction between a user and item, \method\ trains a \textit{predict function} to directly output the embedding of the predicted item that a user will interact with. 
This has the advantage that it \method\ only needs to do one forward pass during inference to generate the item embedding, as opposed to $|\mathcal{I}|$ times (once for each candidate item). 
We illustrate the three major operations of \method\ in Figure~\ref{fig:model}.

\textbf{Static and Dynamic Embeddings.}
In \method , each user and item is assigned two types of embeddings: a static embedding and a dynamic embedding. 

Static embeddings, $\bm{\overline{u}} \in \mathbb{R}^d \text{ } \forall u \in \mathcal{U}$ and $\bm{\overline{i}} \in \mathbb{R}^d \text{ } \forall i \in \mathcal{I}$, do not change over time. 
These are used to express stationary properties such as the long-term interest of users. 
We use one-hot vectors as static embeddings of all users and items, as advised in Time-LSTM~\cite{zhu2017next} and TimeAware-LSTM~\cite{baytas2017patient}. 

On the other hand, each user $u$ and item $i$ is assigned a dynamic embedding represented as $\bm{u(t)} \in \mathbb{R}^n$ and $\bm{i(t)} \in \mathbb{R}^m$ at time $t$, respectively. 
These embeddings change over time to model their evolving behavior. 
The embeddings of a user and item are updated whenever they are involved in an interaction. 

In \method , use both the static and dynamic embeddings to train the model to predict user-item interactions in order to leverage both the long-term and dynamic properties.


\begin{figure}[t]
\centering
        \includegraphics[width=0.8\columnwidth]{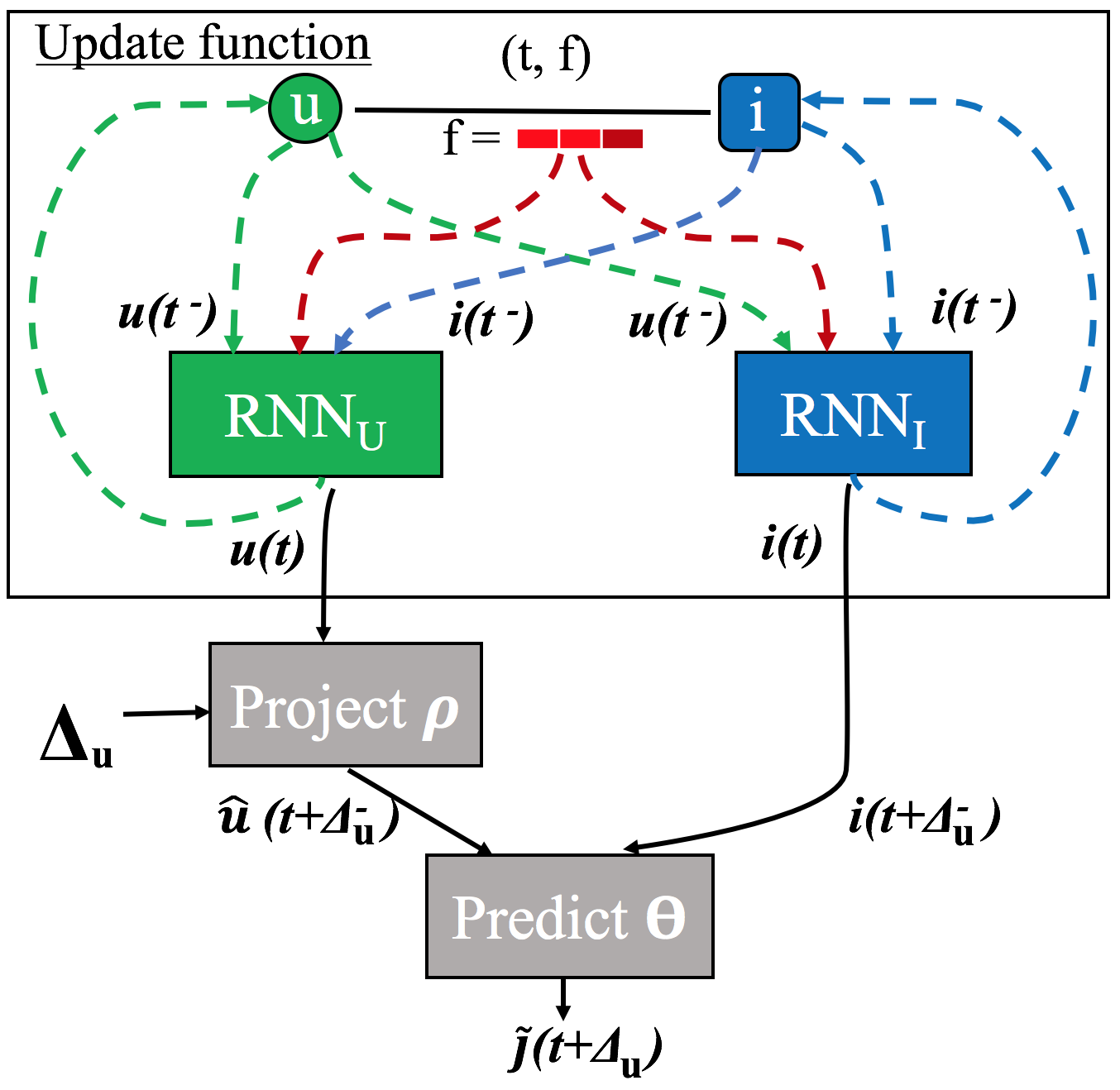}
    \caption{\textbf{Architecture of \method :} After an interaction $S = (u, i, t, f)$ between user $u$ and item $i$, the dynamic embeddings of $u$ and $i$ are updated in the update function with $RNN_U$ and $RNN_I$, respectively. To predict user $u$'s interaction at time $t+\Delta_u$, the user embedding is projected, $\bm{\hat{u}(t+\Delta_u)}$, using the project function $\rho$. This is used to generate the embedding $\bm{\tilde{j}(t+\Delta_u)}$ of the predicted item $j$. 
    \label{fig:model}}
\end{figure}

\subsection{Learning dynamic embeddings with \textit{JODIE}} 

Here we propose a mutually-recursive Recurrent Neural Network based model that learns dynamic embeddings of both users and items jointly. 
We will explain the three major components of the algorithm: update, project, and predict. 
Algorithm~\ref{alg:jodie} shows the algorithm for each epoch.

\subsubsection{\textbf{Update operation using a coupled recurrent model.}} 
\label{sec:update}
In the update operation, the interaction $S = (u, i, t, f)$ between a user $u$ and item $i$ is used to update both their dynamic embeddings. 
Our model uses two separate recurrent neural networks for updates---$RNN_U$ is shared across all users and used to update user embeddings, and $RNN_I$ is shared among all items to update item embeddings. 
The state of the user RNN and item RNN represent the user and item embeddings, respectively.

When user $u$ interacts with item $i$, $RNN_U$ updates the embedding $\bm{u(t)}$ by using the embedding of item $i$ as an input. 
This is in stark contrast to the popular use of items' one-hot vectors to update user embeddings~\cite{beutel2018latent,wu2017recurrent,zhu2017next}, which makes these models infeasible as they scale only to a small number of items due to space complexity. 
Instead, we use the dynamic embedding of an item as it contains more information than just the item's `id', including its current state and its recent interactions with (any) user. Therefore, the use of item's dynamic embeddings can generate more meaningful dynamic user embeddings. 
For the same reason, $RNN_I$ uses the dynamic user embedding to update the dynamic embedding of the item $i$. 
This results in mutually recursive dependency between the embeddings. 
Figure~\ref{fig:model} shows this in the ``update function'' block. 

Formally, 
$$ \bm{u(t)} =  RNN_U(\bm{u(t^-)}, \bm{i(t^-)}, \Delta_u, f) $$ 
$$ \bm{i(t)} =  RNN_I(\bm{i(t^-)}, \bm{u(t^-)}, \Delta_i, f) $$ 
where $\bm{u(t^-)}$ and $\bm{i(t^-)}$ represent the user and item embeddings before the interaction (i.e., those obtained after their previous interaction updates). 
$\Delta_u$ and $\Delta_i$ represent the time elapsed since $u$'s previous interaction (with any item) and $i$'s previous interaction (with any user), respectively, and are used as input to account for their frequency of interaction. 
Incorporating time has been shown to be useful in prior work~\cite{dai2016deep,wu2017recurrent,beutel2018latent,zhang2017deep}. 
The interaction feature vector $f$ is also used as an input. 
The input vectors are all concatenated and fed into the RNNs. 
Variants of RNNs, such as LSTM and GRU, gave empirically similar performance in our experiments, so we use an RNN to reduce the number of trainable parameters.


{
\begin{table}
\caption{\label{tab:symbols} Table of symbols used in this paper.}
\begin{tabular}{|c|l|}
\hline
Symbol & Meaning \\\hline
$\bm{u(t)}$ and $\bm{i(t)}$ & Dynamic embedding of user $u$ and item $i$ at time $t$ \\
$\bm{u(t^-)}$ and $\bm{i(t^-)}$ & Dynamic embedding of user $u$ and item $i$ before time $t$ \\
$\bm{\overline{u}}$ and $\bm{\overline{i}}$ & Static embedding of user $u$ and item $i$ \\
$[\bm{\overline{u}}, \bm{u(t)}]$ and $[\bm{\overline{i}}, \bm{i(t)}]$ & Complete embedding of user $u$ and item $i$ a time $t$\\ 
$RNN_U$ and $RNN_I$ & User RNN and item RNN to update embeddings\\
$\rho$ & Embedding projection function \\
$\Theta$  & Prediction function to output \\
$\bm{\widehat{u}(t)}$ and $\bm{\widehat{i}(t)}$ & Projected embedding of user $u$ and item $i$ at time $t$ \\
$\bm{\widetilde{i}(t)}$ & Predicted item embedding  \\
$f_j$ & Feature at interaction $\mathcal{S}_j$ \\\hline
\end{tabular}
\end{table}
}

\subsubsection{\textbf{Embedding projection operation.}}
\label{sec:project}
Between two interactions of a user, its embedding may become stale as time more time elapses. Using stale embeddings lead to sub-optimal predictions and therefore, it is crucial to estimate the embeddings in real-time. To address this, we create a novel \textit{projection operation} that estimates the embedding of a user after some time $\Delta$ elapses since its previous interaction. 

In practice, consider the scenario when a recommendation needs to be made to a user when it logs into a system. 
For example, on an e-commerce website, if a user returns 5 minutes after a previous purchase, then its projected embedding would be close to its previous embedding. On the other hand, the projected embedding would drift farther if the user returned 10 days later. The use of projected embedding enables \method\ to make different recommendations to the same user at different points in time. 
Therefore, the instantaneous projected embedding of a user can be utilized to make efficient real-time recommendations. The projection operation is one of the major innovations of \method. 

The projection function $\rho: \mathbb{R}^n \times \mathbb{R} \rightarrow \mathbb{R}^n$ projects the embedding of a user after time $\Delta_u$ has elapsed since its previous interaction at time $t$. 
We represent the projected user embedding at time $t + \Delta_u$ as $\bm{\widehat{u}(t+\Delta_u)}$. 

\cut{
The embedding projection operation of \method\ are trained to generate an embedding of a user at time $t + \Delta_u$, where $t$ is the time of its previous interaction. 
 i.e., right before a user $u$ interacts with item $j$, the model should be able to correctly predict `$j$' (e.g., consider predicting which item the user will purchase when it logs in). 
As time $\Delta_u$ has elapsed since $u$'s previous interaction, the state of the user may have changed. 
Therefore, 
}

The two inputs to the projection operation are $u$'s previous embedding at time $t$ and the value $\Delta_u$. 
We follow the method suggested in LatentCross~\cite{beutel2018latent} to incorporate time into the embedding.
We first convert $\Delta_u$ to a time-context vector $w \in \mathbb{R}^n$ using a linear layer, where we initialize $w$ to a 0-mean Gaussian. 
The projected embedding is then obtained as an element-wise product as follows:
$$\bm{\widehat{u}(t + \Delta_u)} = (1 + w) * \bm{u(t)}$$
The time-context vector $w$ essentially acts as an attention vector to scale the past user embedding to the current state. 
The context linear layer is trained during the training phase.

In Figure~\ref{fig:projection}, we show the projected embedding of user $u_2$ in our example network for different values of $\Delta_u$. 
We see that for smaller $\Delta < \Delta_2$, the projected embedding $\bm{\widehat{u}_2(t_7 + \Delta)}$ is closer to the previous embedding $\bm{u_2(t_7)}$, and it drifts farther as the value increases, showing the change in user's state.

\begin{figure}[t!]
\centering
        \includegraphics[width=0.8\columnwidth]{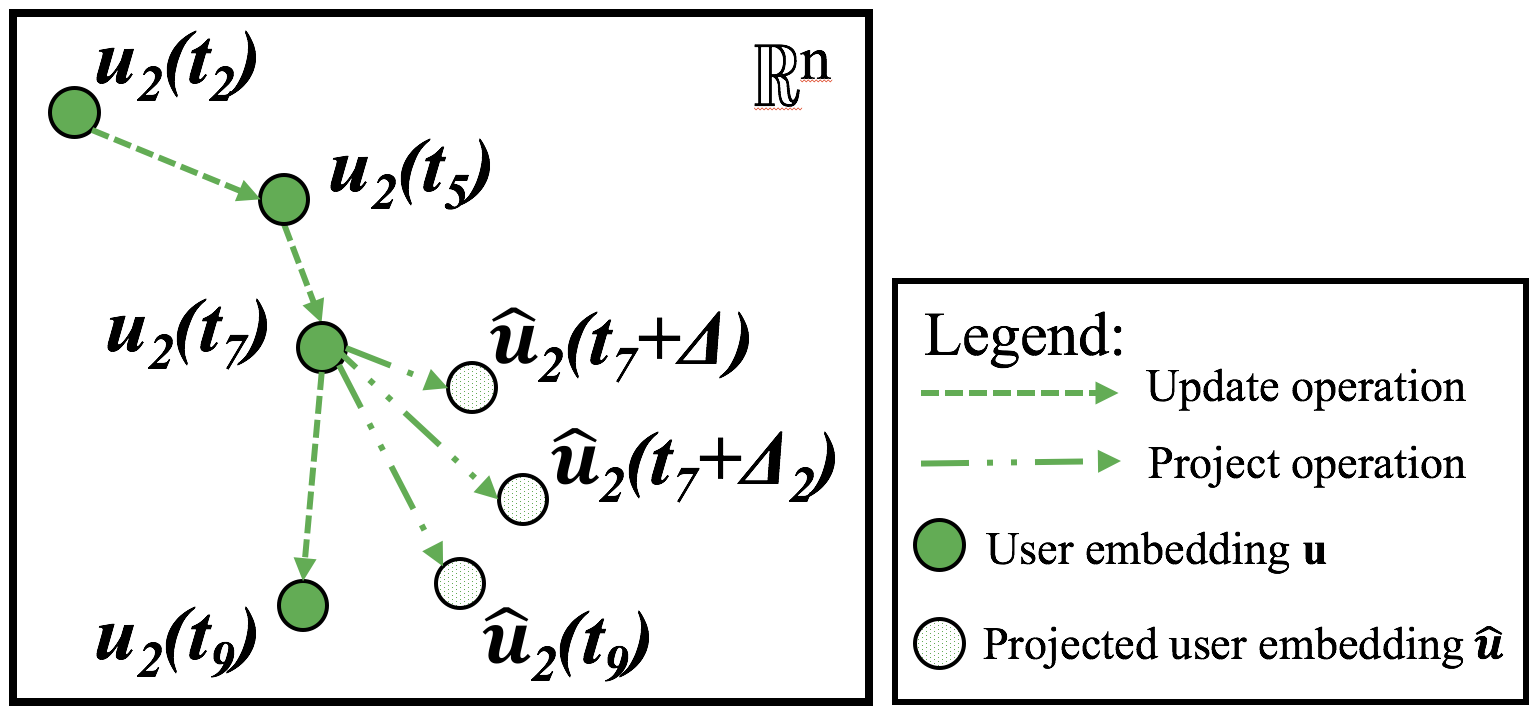}
    \caption{This figure shows the \textit{project operation}. The projected embedding of user $u_2$ in the example network is shown for different time elapsed values $\Delta$ and $\Delta_2 > \Delta$. We see that the projected embedding drifts farther as $\Delta_U$ increases. }
    \label{fig:projection}
\end{figure}

{
\footnotesize
\begin{algorithm}[t] 
    \caption{\textsc{\method\ Algorithm}} 
    \label{alg:jodie} 
	\SetKwInOut{Input}{Input}\SetKwInOut{Output}{Output} 
    \Input{Temporally sorted sequence of interactions $S: S_j = (u, i, t, f)$; \\
    Initial user embeddings $\bm{u(t)}$ $\forall u \in \mathcal{U}$; \\
	Initial item embeddings $\bm{i(t)}$ $\forall i \in \mathcal{I}$; \\
	Current model parameters: $\rho, \Theta, RNN_U, RNN_I$
    } 
    \Output{Dynamic user and item embeddings, updated model parameters} 
    \BlankLine
		$\ell \leftarrow 0$ \;
		$prev \leftarrow \{ \text{ } \} $
		\For{$j \in 1 \text{ to } |S|$} { 
		\tcc{Processing $ (u, i, t, f)$} 
		\tcc{Let $\Delta_u$ be the time since $u$'s previous interaction with any item, and $\Delta_i$ be the time since $i$'s previous interaction with any user}
		
			$\bm{\widehat{u}(t)} \leftarrow \rho(\bm{u(t^-)}, \Delta_u) $ \label{line:idxfind} \tcp*{Project user embedding} 
			$k \leftarrow prev[u]$ \tcp*{$k$ is the previous item $u$ interacted with} 
			$\bm{\widetilde{i}(t)} \leftarrow \Theta(\bm{\widehat{u}(t)}, \bm{\overline{u}}, \bm{k(t^-)}, \bm{\overline{k}} ) $ \label{line:idxfind} \tcp*{Predict item embedding} 
			$\ell \leftarrow \ell + || \bm{i(t^-)} - [\bm{\overline{i}}, \bm{\widetilde{i}(t)}] ||_2$ \label{line:idxfind} \tcp*{Calculate prediction loss} 
		
			\tcc{Update user and item embedding} 
			$ \bm{u(t)} \leftarrow RNN_U(\bm{u(t^-)}, \bm{i(t^-)}, \Delta_u, f)  $ \;
			$ \bm{i(t)} \leftarrow RNN_I(\bm{i(t^-)}, \bm{u(t^-)}, \Delta_i, f)  $ \;
			
			$prev[u] \leftarrow i$ \;

			\tcc{Add user and item embedding drift to loss} 
			$\ell \leftarrow \ell + || \bm{u(t)} - \bm{u(t^-)} ||_2 +   || \bm{i(t)} - \bm{i(t^-)} ||_2 $ \; 
		}   
		$\text{Back-propagate loss and update model parameters}$ 

		\Return{$\bm{u(t)}  \text{ }\forall u \in \mathcal{U} \text{ }  \text{ }\forall t, \bm{i(t)}  \text{ }\forall i \in \mathcal{I} \forall t$, model: $\rho, \Theta, RNN_U, RNN_I$}
\end{algorithm}
}

\subsubsection{\textbf{Predicting user-item interaction.}}
\label{sec:prediction}
The \method\ model is trained to correctly predict future user and item interactions. 
To make prediction of the user's interaction at time $t + \Delta_u$, we introduce a prediction function $\Theta:  \mathbb{R}^{(n+d)} \times \mathbb{R}^{(m+d)} \rightarrow \mathbb{R}^{(m+d)}$. 
This function takes the estimated user embedding along with its static embedding as input. 
Additionally, we also use the static and dynamic embedding of the item $i$ (the item from $u$'s last interaction at time $t$) as inputs. 
As item $i$ could interact with other users between time $t$ and $t + \Delta_u$, its embedding right before time $t + \Delta_u$, i.e., $\bm{i(t+\Delta_u^-)}$, would reflect more recent information. 
We use both the static and dynamic embeddings to predict using both the long-term and temporary properties of the user and items. 
The function is trained to output the complete (static and dynamic) predicted item embedding. 

In practice, we use a fully connected layer as the function $\Theta$. 

Note that instead of predicting a probability score of interaction between a user and a candidate item, \method\ directly predicts an item embedding. 
This is advantageous during inference time (i.e., making real-time predictions), as we would only need to do the expensive neural network forward-pass of the prediction layer once and select the item with the closest embedding using k-nearest neighbor search. This is advantageous over the standard approaches~\cite{dai2016deep,wu2017recurrent,beutel2018latent,DBLP:conf/kdd/DuDTUGS16} that generate a probability score as they need to do the forward pass $N$ times (once for each of the $N$ candidate items) to find the item with the highest probability score.

\subsubsection{\textbf{Loss function for training.}}
The entire \method\ model is trained to minimize the distance between the predicted item embedding and the ground truth item's embedding at every interaction.
Let $u$ interact with item $j$ at time $t+\Delta_u$. 
We calculate the loss as follows: 
$$
Loss =  || \Theta(\bm{\widehat{u}(t + \Delta_u)}, \text{ } \bm{\overline{u}} ,  \text{ } \bm{i(t + \Delta_u^-)}, \text{ } \bm{\overline{i}})  - [\bm{\overline{j}}, \bm{j(t + \Delta_u^-)]} ||_2 $$
$$ \text{ } + \text{ } \lambda_U ||\bm{u(t + \Delta_u)} - \bm{u(t + \Delta_u^-)} ||_2 \text{ } + \text{ }  \lambda_I || \bm{j(t + \Delta_u)} - \bm{j(t + \Delta_u^-)} ||_2 $$

The first loss term minimizes the predicted embedding error. The last two terms are added to regularize the loss and prevent the consecutive dynamic embeddings of a user and item to vary too much, respectively. 
$\lambda_U$ and $\lambda_I$ are scaling parameters to ensure the losses are in the same range. 
Note that we do not use negative sampling during training as \method\ directly outputs the embedding of the predicted item. 

Overall, Algorithm~\ref{alg:jodie} describes the process in each epoch.

\subsubsection{\textbf{Extending the loss term for user state change prediction.} }
In certain prediction tasks, such as user state change prediction, additional training labels may be present for supervision. 
In those cases, we train another prediction function $\Theta: \mathbb{R}^{(n+d)} \rightarrow \mathcal{C}$ to predict the label using the dynamic embedding of the user after an interaction.  
We calculate the cross-entropy loss for categorical labels and add the loss to the above loss function with another scaling parameter. 
We explicitly do not just train to minimize only the cross-entropy loss to avoid overfitting. 

\cut{
\subsubsection{\textbf{Training details.}}
Here we give some important algorithm training details. 

The algorithm is run for several epochs till convergence. 
In every epoch, all embeddings are first initialized, and then the interactions are processed to generate the dynamic embeddings. 
All user embeddings are initialized to the same random vector to set a common `starting point'. 
Another random vector is used to initialize all item embeddings. 
These vectors are trainable so that optimal starting vectors are obtained. 

We train the model on the first $\tau$\% of the total number of interactions. The rest are used for validation and test. 
In each epoch, we split the number of training interactions into batches of equal time duration, so that the model is regularly updated. 
The loss is backpropagated as the end of each batch to frequently update the model parameters.  
Further, we also calculate the loss during testing and backpropagate to actively train the model (i.e., in a prequential setting~\cite{benczur2018online}). 
}

\subsubsection{\textbf{Differences between \method and DeepCoevolve}}
DeepCoevolve is the closest state-of-the-art algorithm as it also trains two joint RNNs for generate dynamic embeddings. 
However, the key differences between \method and DeepCoevolve are the following: (i) \method uses a novel project function to estimate the embedding of a user and item at any time. Instead, DeepCoevolve maintains a constant embedding between two consecutive interactions of the same user/item. This makes \method truly dynamic. (ii) \method predicts the embedding of the next item that a user will interact with. In contrast, DeepCoevolve trains a function to predict the likelihood of interaction between a user and an item. This requires $N$ forward passes through the inference layer (for $N$ items) to select the item with the highest score. On the other hand, \method is scalable at inference time as it only requires one forward pass through the inference layer (see Section~\ref{sec:prediction} for details). 

Overall, in this section, we presented the \method\ algorithm to generate dynamic embeddings for users and items from a sequence of user-item interactions.

\section{t-Batch: Time-Consistent Batching Algorithm}
\label{sec:batching}

Here we explain a general batching algorithm to parallelize the training of models that learn from a sequence of user-item interactions. 

Training two mutually-recursive recurrent networks, in models such as DeepCoevolve~\cite{dai2016deep} and \method, introduces new challenges as it is fundamentally different from training a single RNN. The standard RNN models are trained through the standard Back Propagation Through Time (BPTT) mechanism. 
Methods such as RRN~\cite{beutel2018latent} and T-LSTM~\cite{zhu2017next} that use single RNNs can split users into small batches because each user's interaction sequence is treated independently (using one-hot vector representations of items). This enables parallelism but ignores the user-to-user interdependencies. 
On the other hand, in \method\ and DeepCoevolve, the two RNNs are mutually recursive to incorporate user-to-user dependencies, and as a result, users' interactions sequences are not independent. 
So, the two RNNs can not be trained independently. 
This new challenge requires new methods for efficient training. 

{
\footnotesize
\begin{algorithm}[t] 
    \caption{\textsc{\batching Algorithm}} 
    \label{alg:batching} 
	\SetKwInOut{Input}{Input}\SetKwInOut{Output}{Output} 
    \Input{Temporally sorted sequence of interactions $S: S_j = (u_j, i_j, t_j, f_j)$} 
    \Output{Sequence of batches $B: B_k = \{S_{k1}, \ldots S_{kn}\}$} 
    \BlankLine
		$last_U[u] \leftarrow 0 \forall u \in \mathcal{U}$  \tcp*{Initialize all users' last-batch index} \label{line:userlastinit} 
		$last_I[i] \leftarrow 0 \forall i \in \mathcal{I}$ \tcp*{Initialize all items' last-batch index} \label{line:itemlastinit} 
		$B_k \leftarrow \{\} \for k \in [1, |S|]$ \label{line:batchinit} \tcp*{Initialize batches} 
		$C \leftarrow 0$  \label{line:countinit} \tcp*{Initialize non-empty batch count}
	\For{$j \in 1 \text{ to } |S|$} { 
		\tcc{Processing $S_j =  (u_j, i_j, t_j, f_j)$} 
		$idx \leftarrow max (last_U[u_j], last_I[i_j] ) + 1 $  \label{line:idxfind} \tcp*{Find batch id to insert} 
		$B_{idx} \leftarrow B_{idx} \cup \{S_j\}$ \label{line:batchaddition}  \tcp*{Add interaction to correct batch}
		$last_u[u_j] \leftarrow idx$ \label{line:useridxupdate} \tcp*{Update user last-batch index}
		$last_I[i_j] \leftarrow idx$ \label{line:itemidxupdate} \tcp*{Update item last-batch index}
		$C \leftarrow max(C, idx)$ \tcp*{Update batch count $C$}
	}   
	\Return{$\{B_1, \ldots B_C\}$}
\end{algorithm}
}

First we propose two necessary conditions for any batching algorithm to work with coupled recurrent networks. 
\begin{conditions}[Co-batching conditions]
A batching algorithm for coupled recurrent networks should satisfy the following two conditions: 
\begin{enumerate}
\item Every user and every item can only appear at most once in a batch. This is required so that all interactions in the batch are completely independent so that the batch can be parallelized, and 
\item The $k^{th}$ and $k+1^{st}$ interactions of a user or item should be assigned to batches $B_i$ and $B_j$, respectively, such that $i < j$ . This is required since batches are processed sequentially, $u$'s $k^{th}$ embedding can then be used in the $k+1^{st}$ interaction. 
\end{enumerate}
\end{conditions}


The naive solution that satisfies both the above conditions is to process interactions one-at-a-time. 
However, this is very slow and can not be scaled to large number of interactions. 
This approach is used in existing methods such as Dai et al.~\cite{dai2016deep} and Zhang et al.~\cite{zhang2017learning}. 

Therefore, we propose a novel batching algorithm, called \batching, that creates large batches that can be parallelized for faster training. 
\batching is shown in Algorithm~\ref{alg:batching}. 
\batching takes as input the temporally-sorted sequence of interactions $S$ and outputs a set of batches. 
It processes one interaction at a time in increasing order of time. 
The key idea of \batching is that the $j^{th}$ interaction (say, between user $u$ and $i$) is assigned to the batch id after the largest batch id with any interaction of $u$ or $i$ yet. 
As a result, each batch only has unique user and item, and the batches can be processed in increasing order of their id, satisfying both the co-batching conditions. 

Algorithmically, the variables $last_U$ and $last_I$ store the index of the last batch in which each user and item appears, respectively (lines~\ref{line:userlastinit} and \ref{line:itemlastinit}).
A total of $|S|$ empty batches are initialized which is the maximum number of resulting batches (line ~\ref{line:batchinit}) and a count $C$ of the number of non-empty batches is maintained (line~\ref{line:countinit}).
All interactions $S_j \in S$ are then sequentially processed, where the interaction $S_j$ is added to the batch after the one in which either $u_j$ or $i_j$ appear. 
All non-empty batches are returned at the end.

\begin{example}[Running example]
In the example interaction network shown in Figure~\ref{fig:toy}, \batching results in the following batches:
\begin{itemize}
\item Batch $B_1$: $(u_1, i_1, t_1, f_1), (u_3, i_2, t_3, f_3)$
\item Batch $B_2$: $(u_2, i_1, t_2, f_2), (u_1, i_2, t_4, f_4), (u_3, i_3, t_6, f_6)$
\item Batch $B_3$: $(u_2, i_2, t_5, f_5), (u_3, i_4, t_8, f_8)$
\item Batch $B_4$: $(u_2, i_3, t_7, f_7)$
\item Batch $B_5$: $(u_2, i_4, t_8, f_8)$
\end{itemize}
We can see that a total of 5 batches, there is a 45\% decrease compared to the naive 9 batches. 
Note that in each batch, users and items appear at most once, and for the same user (and item), earlier transactions are assigned earlier batches. 
\end{example}

\begin{theorem}
\batching algorithm satisfies the co-batching conditions. 
\end{theorem}
\begin{proof}
Lines~\ref{line:idxfind} and \ref{line:batchaddition} ensure that every batch contains a user and an item only once. This satisfies condition 1.

The interactions are added to batches in increasing order of time. 
Consider two interactions $S_j$ and $S_k (k > j)$ in which a user $u$ appears. 
If $S_j$ is added to batch $n$, the last-batch index of $u$ is set of $n$, such that when $S_k$ is processed, the index is at least $n+1$ (as shown in line~\ref{line:userlastinit}). 
This satisfies condition 2 and completes the proof.
\end{proof}

\begin{theorem}[Complexity]
The complexity of \batching in creating the batches is $\mathcal{O}(|S|)$, i.e., linear in the number of interactions, as each interaction is seen only once. 
\end{theorem}

Overall, in this section, we presented the \batching algorithm that creates batches from training data such that each batch can be parallelized. This leads to faster training of large scale interaction data. In Section~\ref{sec:exp-tbatch}, we experimentally validate that \batching\ leads to a speed-up between 7.4--8.5 $\times$ in \method\ and DeepCoevolve.

\cut{
Here we create a batching algorithm that maximally creates batches such that every user and every item can only appear at most once in a batch. 
This way all interactions in the batch are completely independent of one another and they can be parallelized. 
Moreover, the batches are time-consistent such that if a user (or item) appears in batch $b_i$ and $b_j$ such that $i < j$, then the time-stamp of interaction of the user (or item) in the two batches is strictly $b_i(u) <  b_j(u)$. 
This way updated embeddings from $u$'s $k^{th}$ interaction is used in its $k+1^{st}$ interaction and the temporal ordering is strictly maintained. 
}

\cut{ 
Now we explain the batching process. 
Thinking of the interaction network, the batches are created by maximally selecting edges from the network.
Starting from the entire network, for each node, its incident edge with the minimum timestamp is selected. 
If multiple edges are selected for a node (because of its selection from neighbors), then the one with the minimum timestamp is selected. 
Multiple edges are removed till every node only has at most one edge selected. 
The selected edges are added to the next batch, and removed from the network. 
This process is repeated till all the edges have been assigned to a batch. 
}


\section{Experiments}
\label{sec:experiments}


In this section, we experimentally validate the effectiveness of \method\ on two tasks: next interaction prediction and user state change prediction. 
We conduct experiments on three datasets each and compare with six strong baselines to show the following: 
\begin{enumerate}
\item \method\ outperforms the best performing baseline by up to 22.4\% in predicting the next interaction and up to 4.5\% in predicting label changes. 
\item  We show that \batching results in over 7.4$\times$ speed-up in the running-time of both \method\ and DeepCoevolve. 
\item  \method\ is robust in performance to the availability of training data. 
\item We show that the performance of \method\ is stable with respect to the dimensionality of the dynamic embedding.
\item Finally, we show the usefulness of \method\ as an early-warning system for label change. 
\end{enumerate}

We first explain the experimental setting and the baseline methods, and then illustrate the experimental results. 

\textbf{Experimental setting.}
We train all models by splitting the data by time, instead of splitting by user which would result in temporal inconsistency between training and test data. 
Therefore, we train all models on the first $\tau$ fraction of interactions, validate on the next $\tau_v$ fraction, and test on the next $\tau_t$ fraction of interactions.

For fair comparison, we use 128 dimensions as the dimensionality of the dynamic embedding for all algorithms and one-hot vectors for static embeddings. 
All algorithms are run for 50 epochs, and all reported numbers for all models are for the test data corresponding to the best performing validation set.

\textbf{Baselines.}
We compare \method\ with six state-of-the-art algorithms spanning three algorithmic categories: 
\begin{enumerate}
\item \textbf{Recurrent neural network algorithms:} in this category, we compare with RRN~\cite{wu2017recurrent}, LatentCross~\cite{beutel2018latent}, Time-LSTM~\cite{zhu2017next}, and standard LSTM. These algorithms are state-of-the-art in recommender systems and generate dynamic user embeddings. We use Time-LSTM-3 cell for Time-LSTM as it performs the best in the original paper~\cite{zhu2017next}, and LSTM cells in RRN and LatentCross models. As is standard, we use the one-hot vector of items as inputs to these models. 
\item \textbf{Temporal network embedding algorithms:} we compare \method\ with CTDNE~\cite{nguyen2018continuous} which is the state-of-the-art in generating embeddings from temporal networks. As it generates static embeddings, we generate new embeddings after each edge is added. We use uniform sampling of neighborhood as it performs the best in the original paper~\cite{nguyen2018continuous}. 
\item \textbf{Co-evolutionary recurrent algorithms:} here we compare with the state-of-the-art algorithm, DeepCoevolve~\cite{dai2016deep}, which has been shown to outperform other co-evolutionary point-process algorithms~\cite{trivedi2017know}. We use 10 negative samples per interaction for computational tractability. 
\end{enumerate}


{
\begin{table}
\small
\caption{\label{tab:interaction} \textbf{Future interaction prediction experiment:} Table comparing the performance of \method\ with state-of-the-art algorithms, in terms of mean reciprocal rank (MRR) and recall@10.  The {\color{blue!75}best algorithm} in each column is colored {\color{blue!75}blue} and {\color{blue!20}second best is light blue}. \method\ outperforms the baselines by up to 22.4\%.}
\begin{tabular}{c|c|c|c|c|c|c}
\hline
Method & \multicolumn{2}{c|}{Reddit} &  \multicolumn{2}{|c|}{Wikipedia} & \multicolumn{2}{|c}{LastFM} \\
& MRR & Rec@10 & MRR & Rec@10 & MRR & Rec@10 \\\hline 
LSTM~\cite{zhu2017next} & 0.355 & 0.551 & 0.329 & 0.455 & 0.062 & 0.119 \\
Time-LSTM~\cite{zhu2017next} & 0.387 & 0.573 & 0.247 & 0.342 & 0.068 & 0.137 \\
RRN~\cite{wu2017recurrent} & \cellcolor{blue!10}0.603 & \cellcolor{blue!10}0.747 & \cellcolor{blue!10}0.522 &\cellcolor{blue!10} 0.617 & 0.089 & 0.182 \\
LatentCross~\cite{beutel2018latent} & 0.421 & 0.588 & 0.424 & 0.481 & \cellcolor{blue!10}0.148 & \cellcolor{blue!10}0.227 \\
CTDNE~\cite{nguyen2018continuous} & 0.165 & 0.257 & 0.035 & 0.056 & 0.01 & 0.01 \\
DeepCoevolve~\cite{dai2016deep} & 0.171 & 0.275 & 0.515 & 0.563 & 0.019 & 0.039 \\\hline 
\method\ (proposed) &  \cellcolor{blue!25}\textbf{0.726} &  \cellcolor{blue!25}\textbf{0.852} &  \cellcolor{blue!25}\textbf{0.746} &  \cellcolor{blue!25}\textbf{0.822} &  \cellcolor{blue!25}\textbf{0.195} &  \cellcolor{blue!25}\textbf{0.307} \\\hline
\% Improvement & 12.3\% & 10.5\% & 22.4\% & 20.5\% & 4.7\% & 8.0\% \\\hline
\end{tabular}
\end{table}
}
\subsection{Experiment 1: Future interaction prediction}
\label{sec:exp1}
In this experiment, the task is to predict future interactions. 
The prediction task is: given all interactions till time $t$, and the user $u$ involved in the interaction at time $t$, which item will $u$ interact with (out of all $N$ items)?

We use three datasets in the experiments related to future interaction prediction:\\
$\bullet$ \textbf{Reddit post dataset:} this dataset consists of one month of posts made by users on subreddits~\cite{pushshift}. We selected the 1000 most active subreddits as items and the 10,000 most active users. This results in 672,447 interactions. We convert the text of the post into a feature vector representing their LIWC categories~\cite{pennebaker2001linguistic}. \\
$\bullet$ \textbf{Wikipedia edits:} this dataset is one month of edits made by edits on Wikipedia pages~\cite{wikidump}. We selected the 1000 most edited pages as items and editors who made at least 5 edits as users (a total of 8227 users). This generates 157,474 interactions. Similar to the Reddit dataset, we convert the edit text into a LIWC-feature vector. \\
$\bullet$ \textbf{LastFM song listens:} this dataset has one months of who-listens-to-which song information~\cite{lastfm}. We selected all 1000 users and the 1000 most listened songs resulting in 1293103 interactions. In this dataset, interactions do not have features. 

We select these datasets such that they vary in terms of users' repetitive behavior: in Wikipedia and Reddit, a user interacts with the same item consecutively in 79\% and 61\% interactions, respectively, while in LastFM, this happens in only 8.6\% interactions. 

\textbf{Experimentation setting.}
We use the first 80\% data to train, next 10\% to validate, and the final 10\% to test. 
We measure the performance of the algorithms in terms of the mean reciprocal rank (MRR) and recall@10---MRR is the average of the reciprocal rank and recall@10 is the fraction of interactions in which the ground truth item is ranked in the top 10. Higher values for both are better. 
For every interaction, the ranking of ground truth item is calculated with respect to all the items in the dataset. 

\textbf{Results.}
Table~\ref{tab:interaction} compares the results of \method\ with the six baseline methods. 
We observe that \method\ significantly outperforms all baselines in all datasets across both metrics on the three datasets (between 4.7\% and 22.4\%).
Interestingly, we observe that our model performs well irrespective of how repetitive users are---it achieves up to 22.4\% improvement in Wikipedia and Reddit (high repetition), and up to 8\% improvement in LastFM . This means \method\ is able to learn to balance personal preference with users' non-repetitive interaction behavior. 
Moreover, among the baselines, there is no clear winner---while RRN performs the better in Reddit and Wikipedia, LatentCross performs better in LastFM. 
As CTDNE generates static embedding, its performance is low. 

Overall, \method\ outperforms these baselines by learning efficient update, project, and predict functions.

{
\small
\begin{table}
\caption{\label{tab:churn} \textbf{User state change prediction:} Table comparing the performance in terms of AUC of \method\ with state of the art algorithms. The {\color{blue!75}best algorithm} in each column is colored {\color{blue!75}blue} and {\color{blue!20}second best is light blue}. \method\ outperforms the baselines by up to 4.5\%.}
\begin{tabular}{c|c|c|c}
\hline
Method & Reddit & Wikipedia & MOOC \\\hline
LSTM~\cite{zhu2017next} & 0.523 & 0.575 & 0.686 \\
Time-LSTM~\cite{zhu2017next} & 0.556 & 0.671 & \cellcolor{blue!10}0.711 \\
RRN~\cite{wu2017recurrent} & \cellcolor{blue!10}0.586 & \cellcolor{blue!10}0.804 & 0.558 \\
LatentCross~\cite{beutel2018latent} & 0.574 & 0.628 & 0.686 \\
DeepCoevolve~\cite{dai2016deep} & 0.577 & 0.663 & 0.671 \\\hline 
\method (proposed method)& \cellcolor{blue!25}\textbf{0.599} &\cellcolor{blue!25} \textbf{0.831} & \cellcolor{blue!25}\textbf{0.756} \\\hline
Improvement over best baseline & 1.3\% & 2.7\% & 4.5\% \\\hline
\end{tabular}
\end{table}
}

\subsection{Experiment 2: User state change prediction}
\label{sec:exp2}
In this experiment, the task is to predict if an interaction will lead to a change in user, particularly in two use cases: predicting banning of users and predicting if a student will drop-out of a course. 
Till a user is banned or drops-out, the label of the user is `0', and their last interaction has the label `1'. For users that are not banned or do not drop-out, the label is always `0'.
This is a highly challenging task because of very high imbalance in labels. 

We use three datasets for this task: \\
$\bullet$ \textbf{Reddit bans:} we augment the Reddit post dataset (from Section~\ref{sec:exp1}) with ground truth labels of banned users from Reddit.
This gives 366 true labels among 672,447 interactions (= 0.05\%).\\
$\bullet$ \textbf{Wikipedia bans:} we augment the Wikipedia edit data (from Section~\ref{sec:exp1}) with ground truth labels of banned users~\cite{wikidump}. This results in 217 positive labels among 157,474 interactions (= 0.14\%).  \\
$\bullet$ \textbf{MOOC student drop-out:} this dataset consists of actions, e.g., viewing a video, submitting answer, etc., done by students on a MOOC online course~\cite{kddcup}. This dataset consists of 7047 users interacting with 98 items (videos, answers, etc.) resulting in over 411,749 interactions. There are 4066 drop-out events (= 0.98\%).

\textbf{Experimentation setting.}
Due to sparsity of positive labels, in this experiment we train the models on the first 60\% interactions, validate on the next 20\%, and test on the last 20\% interactions. 
We evaluate the models using area under the curve metric (AUC), a standard metric in these tasks with highly imbalanced labels. 

For the baselines, we train a logistic regression classifier on the training data using the dynamic user embedding as input. 
As always, for all models, we report the test AUC for the epoch with the highest validation AUC. 

\textbf{Results.}
Table~\ref{tab:churn} compares the performance of \method\ on the three datasets with the baseline models. 
We see that \method\ outperforms the baselines by up to 2.7\% in the ban prediction task and by 4.5\% in the drop-out prediction task. 
As before, there is no clear winner among baselines---RRN performs the second best in predicting bans on Reddit and Wikipedia, while Time-LSTM is the second best in predicting dropouts. 

Thus, \method\ is highly efficient in both link prediction and label change prediction. 

\begin{figure}[t]
\centering
        \includegraphics[width=0.7\columnwidth]{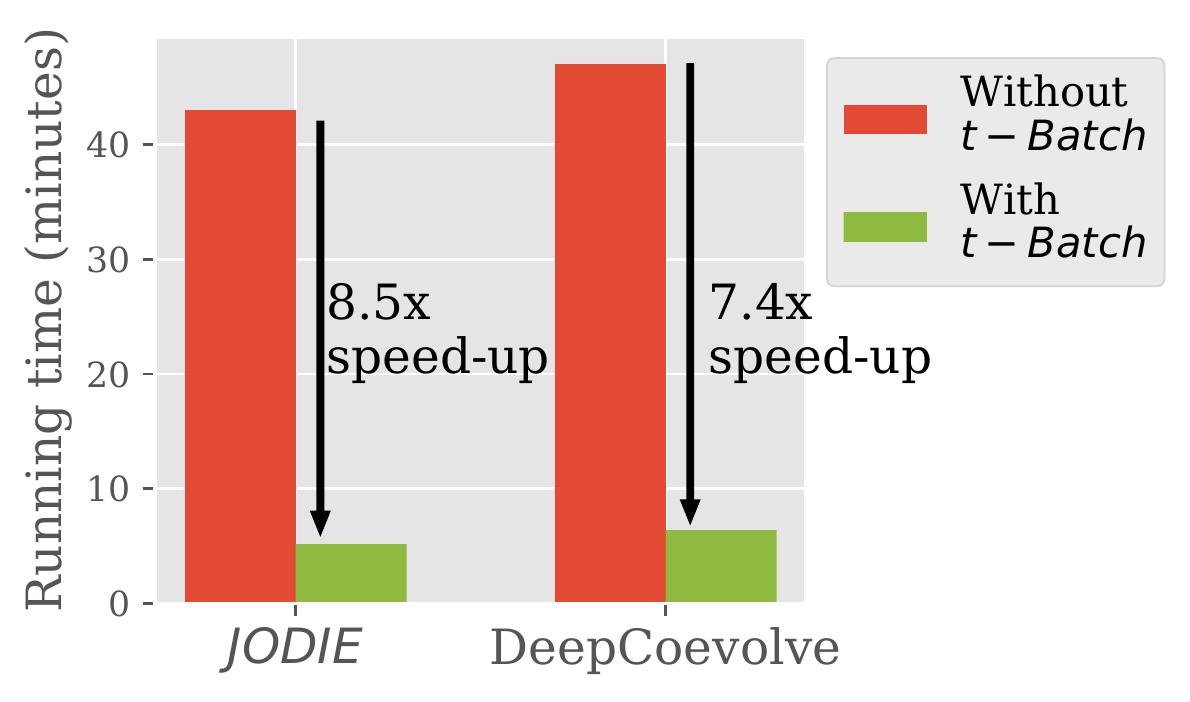} 
    \caption{Figure showing the running time (in minutes) of \method\ and DeepCoevolve, both with and without using the proposed \batching algorithm. \batching speeds both the algorithms by 8.5$\times$ and 7.4$\times$, respectively.  \label{fig:tbatch}}
\end{figure}

\begin{figure*}[t]
\centering
\subfigure[\vspace{-3mm}]{
        \includegraphics[width=0.22\textwidth]{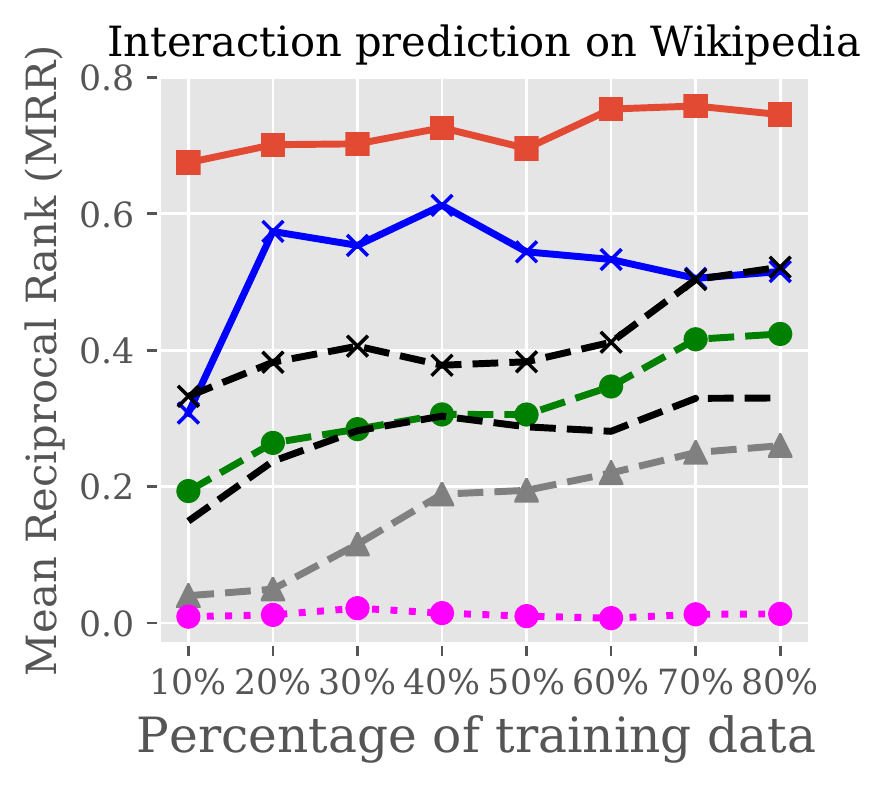}
        }
\subfigure[\vspace{-3mm}]{
        \includegraphics[width=0.21\textwidth]{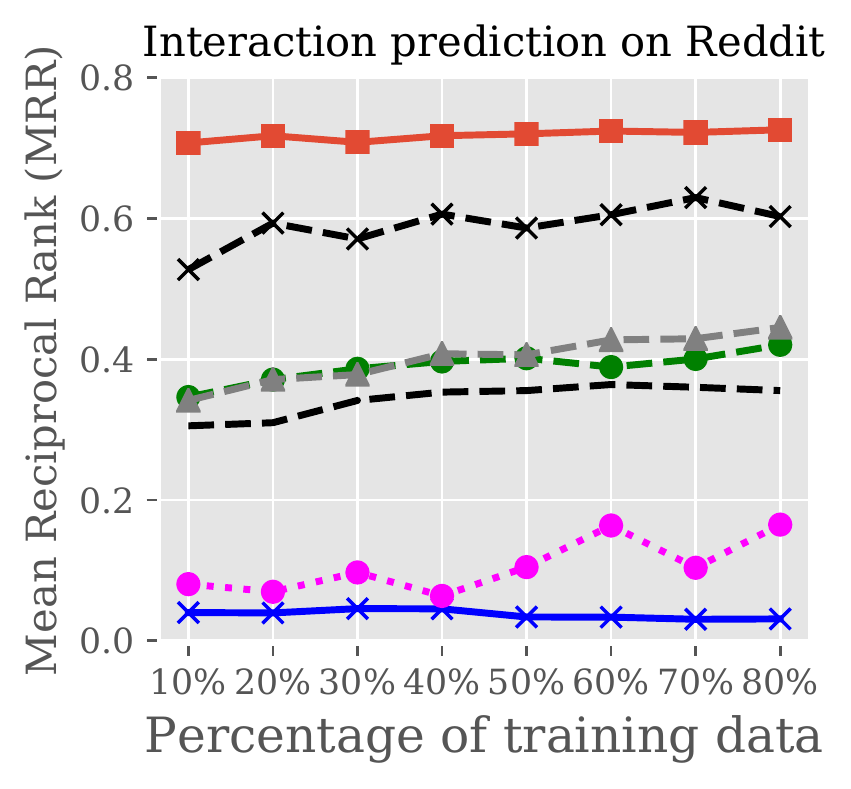}
        }
\subfigure[\vspace{-3mm}]{
                \includegraphics[width=0.22\textwidth]{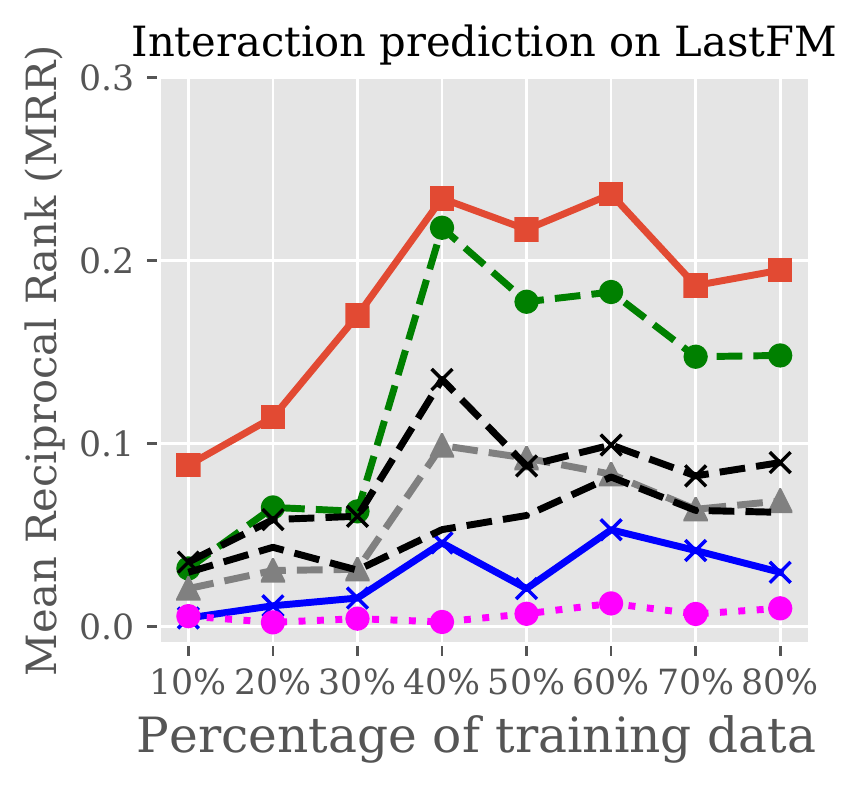}     
        }
\subfigure[\vspace{-3mm}]{
                 \includegraphics[width=0.24\textwidth]{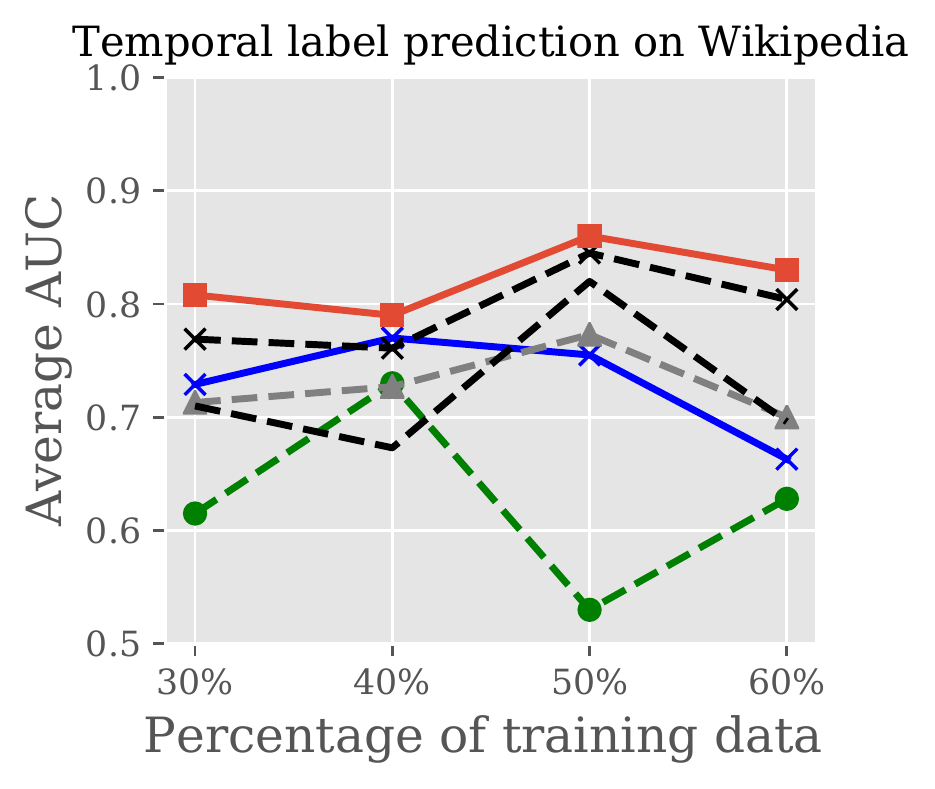}       
                }
                
                             \includegraphics[width=0.8\textwidth]{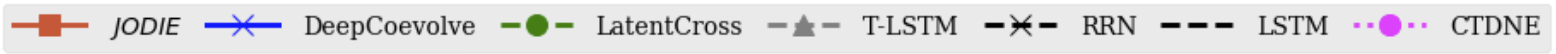} 
    \caption{\textbf{Robustness of \method :} Figures (a--c) compare the mean reciprocal rank (MRR) of \method\ with baselines on interaction prediction task, by varying the training data size. Figure (d) shows the AUC of user state change prediction task by varying the training data size. In all cases, \method\ is consistently the best by up to 33\%. \label{fig:interaction}}
\end{figure*}

\begin{table}
\caption{Table comparing the running time (in minutes) of the two coupled recurrent models, \method and DeepCoevolve, showing the effectiveness of the proposed \batching method. Experiments were conducted on the Reddit dataset. First, we observe that \method is slightly faster than DeepCoevolve and second, we observe that \batching leads to 7.4$\times$--8.5$\times$ reduction in running time.  \label{tab:runtime}}
\begin{tabular}{c|c|c}
\hline
& Without \batching & With \batching \\\hline
DeepCoevolve & 47.21 & 6.35 \\
\method & 43.53 & 5.13 \\\hline
\end{tabular}

\end{table}

\subsection{Experiment 3: Effectiveness of \batching}
\label{sec:exp-tbatch}
Here we empirically show the advantage of \batching algorithm on co-evolving recurrent models, namely our proposed \method\ and DeepCoevolve. 
Figure~\ref{fig:tbatch} shows the running time (in minutes) of one epoch of the Reddit dataset.\footnote{We ran the experiment on one NVIDIA Titan X Pascal GPUs with 12Gb of RAM at 10Gbps speed.}

We make three crucial observations. 
First, we observe that \method\ is slightly faster than DeepCoevolve. We attribute it to the fact that \method\ does not use negative sampling while training because it directly generates the embedding of the predicted item. In contrast, DeepCoevolve requires training with negative sampling. 
Second, we see that our proposed \method\ + \batching combination is 9.2$\times$ faster than the DeepCoevolve algorithm. 
Third, we observe that \batching speeds-up the running-time of both \method\ and DeepCoevolve by 8.5$\times$ and 7.4$\times$, respectively. 

Altogether, these experiments show that the proposed \batching is very effective in creating parallelizable batches from complex temporal dependencies that exist in user-item interactions.
Moreover, this also shows that \batching is general and applicable to algorithms that learn from sequence of interactions.



%
%

\subsection{Experiment 4: Robustness to training data}
\label{sec:exp-training} 
In this experiment, we check the robustness of \method\ by varying the percentage of training data and comparing the performance of the algorithms in both the tasks of interaction prediction and user state change prediction.

For interaction prediction, we vary the training data percentage from 10\% to 80\%. In each case, we take the 10\% interactions after the training data as validation and the next 10\% interactions next as testing. This is done to compare the performance on the same testing data size. 
Figure~\ref{fig:interaction}(a--c) shows the change in mean reciprocal rank (MRR) of all the algorithms on the three datasets, as the training data size is increased. 
We note that the performance of \method\ is stable as it does not vary much across the data points. 
Moreover, \method\ consistently outperforms the baseline models by a significant margin (by a maximum of 33.1\%). 

Similar is the case in user state change prediction. Here, we vary training data percent as 20\%, 40\%, and 60\%, and in each case take the following 20\% interactions as validation and the next 20\% interactions as test. 
Figure~\ref{fig:interaction}(d) shows the AUC of all the algorithms on the Wikipedia dataset. We omit the other datasets due to space constraints, which have similar results. 
Again, we observe that \method\ is stable and consistently performs the best (better by up to 3.1\%), irrespective of the training data size. 

This shows the robustness of \method\ to the amount of available training data.

\begin{figure}[t]
\centering
        \includegraphics[width=0.5\columnwidth]{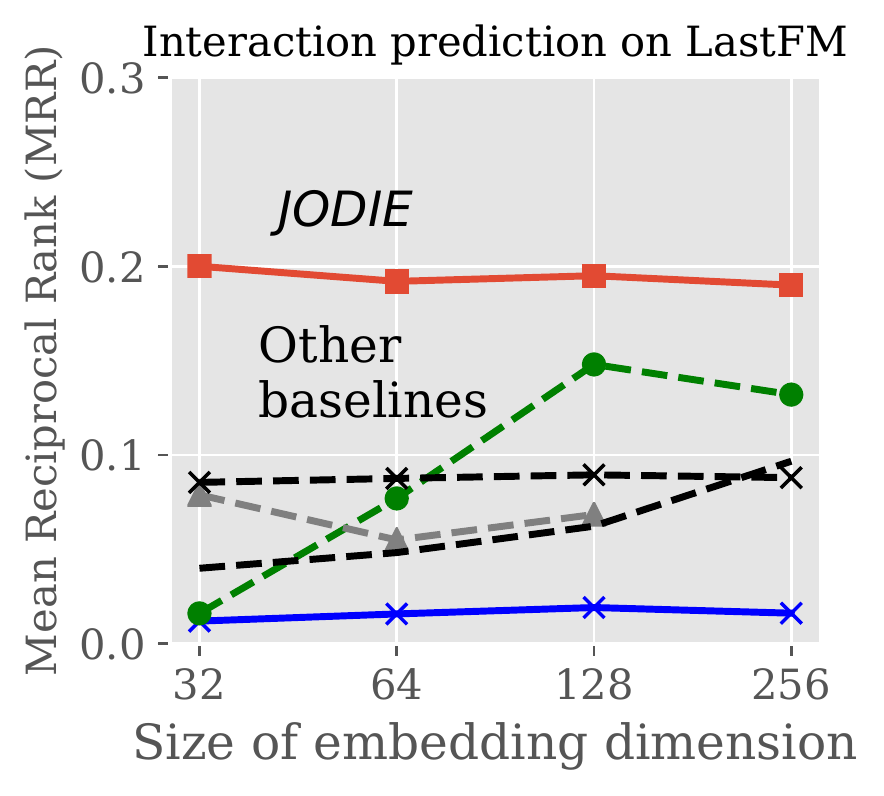} 
    \caption{\textbf{Robustness to dynamic embedding size:} Figure shows that the MRR of \method\ is stable with change in dynamic embeddign size, for the task of interaction prediction on LastFM dataset. Please refer to the legend in Figure 5. \label{fig:embed}
            }
\end{figure}

\subsection{Experiment 5: Robustness to embedding size} 
\label{sec:exp-embsize}
Finally, we check the effect of the dynamic embedding size on the predictions. 
To do this, we vary the dynamic embedding dimension from 32 to 256, and calculate the mean reciprocal rank for interaction prediction on the LastFM dataset. The effect on other datasets is similar and omitted due to space constraints. 
The resulting figure is showing in Figure~\ref{fig:embed}.
We find that the embedding dimension size has little effect on the performance of \method\ and it performs the best overall. 

\begin{figure}[t]
\centering
        \includegraphics[width=0.8\columnwidth]{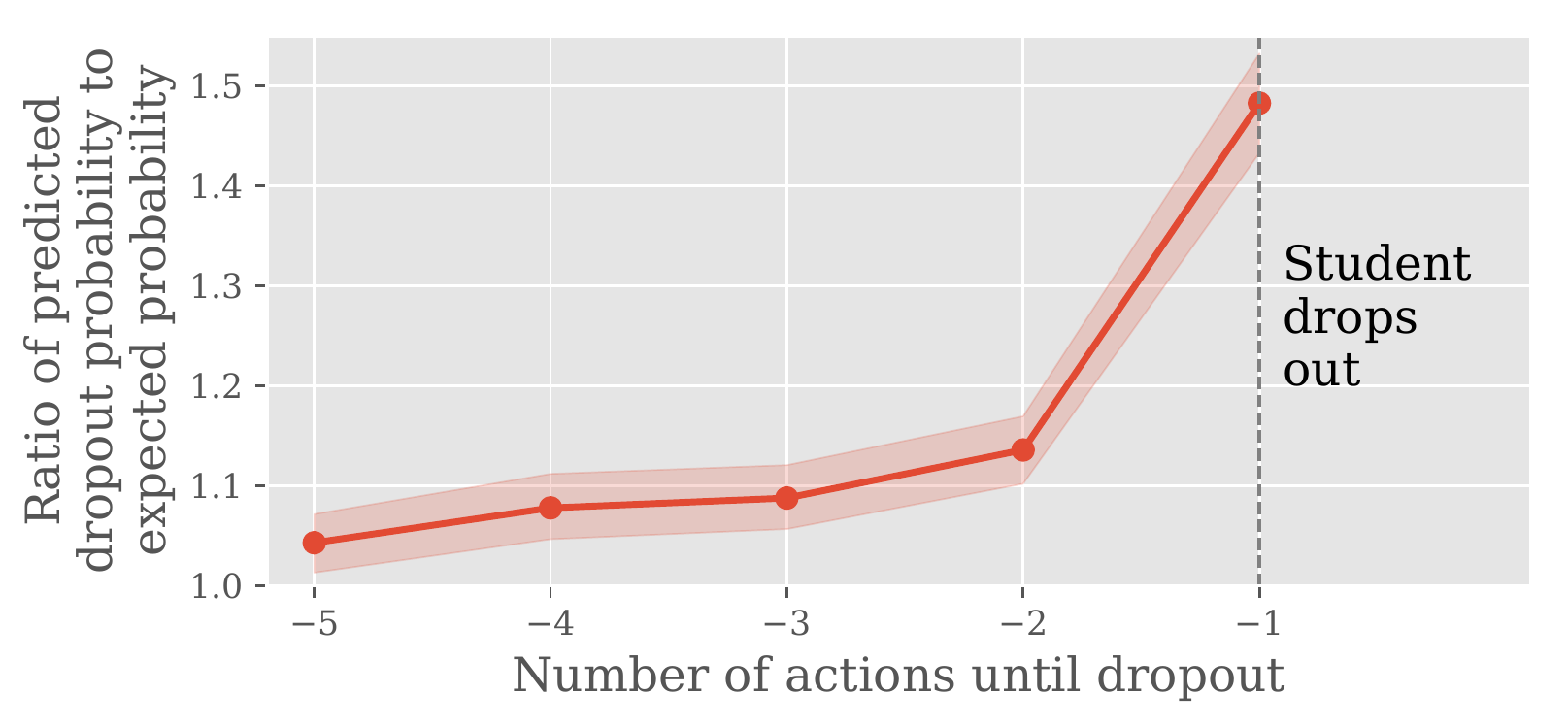} 
    \caption{\textbf{\method\ as an early-warning system:} Figure showing that as a student gets closer to dropping-out (moving right), \method\ predicts a higher dropping out probability score for them compared to other students. 
    \label{fig:ews}}
\end{figure}

\subsection{Experiment 6: \method\ as an early-warning system}
\label{sec:exp-early}
In user state change prediction tasks such as predicting student drop-out from courses and finding online malicious users, it is crucial to make the predictions early, in order to develop effective intervention strategies~\cite{}. 
For instance, if it can be predicted well in advance that a student is likely to drop a course (an `at-risk' student), then steps can be taken by the teachers to ensure continued education of the student~\cite{}.
Therefore, here we show that \method\ is effective in making early predictions for at-risk students. 

To measure this, we calculate the change in student's dropout probability predicted by \method\ as a function of the number of interactions till it drops out. 
Let us call the set of students who drop as $D$ and those who do not as $\overline{D}$. 
We plot the ratio of the predicted probability score of $d \in D$ to the predicted score for $\overline{d} \in \overline{D}$ (i.e., the expected score). 
A ratio of one means that the algorithm gives equal score to dropping-out and non-dropping-out students, while a ratio greater than one means that the algorithm gives a higher score to students that drop out compared to the students that do not. 
The average ratio is shown in Figure~\ref{fig:ews}, with 95\% confidence intervals. 
Here we only consider the interactions that occur in the test set, to prevent direct training on the drop-out interactions. 

First, we observe from Figure~\ref{fig:ews} that the ratio score is higher than one as early as five interactions prior to the student dropping-out.
Second, we observe that as the student approaches its `final' drop-out interaction, its score predicted by \method\ increases steadily and spikes strongly at its final interaction. 
Both these observations together shows that \method\ identifies early signs of dropping-out and predicts a higher score for these at-risk students. 

We make similar observation for users before they get banned on Wikipedia and Reddit, but the ratio in both these cases is close to 1, indicating that there is low early predictability of when a user will be banned. 



\vspace{2mm}
Overall, in this section, we showed that effectiveness and robustness of \method\ and \batching in two tasks, in comparison to six state-of-the-art algorithms. 
Moreover, we showed the usefulness of \method\ as an early-warning system to identify student dropouts as early as five interactions prior to dropping-out.



\section{Conclusions}
We proposed a coupled recurrent neural network model called \method\ that learns dynamic embeddings of users and items from a sequence of temporal interactions. 
The use of a novel project function, inspired by Kalman Filters, to estimate the user embedding at any time point is a key innovation of \method\ and leads to the advanced performance of \method .
We also proposed the \batching algorithm that creates parallelizable batches of training data, which results in massive speed-up in running time. 
\cut{
\method\ uses a novel project function, inspired by Kalman Filters, to estimate the user embedding at any time point. 
We also proposed the \batching algorithm that can create batches during training that enable learning in parallel. 
This makes \method\ highly scalable and 9.2$\times$ faster than existing methods. 
We evaluated \method\ on two tasks---future interaction prediction and user state change prediction---using four real-world datasets. 
We showed that \method\ outperforms six start-of-the-art algorithms in both these tasks significantly. 
}

There are several directions open for future work, such as learning embeddings of groups of users and items in temporal interactions and learning hierarchical embeddings of users and items. We will explore these directions in future work. 


%

\bibliographystyle{abbrv}
\balance
\bibliography{bib/refs,bib/reviewer_refs,bib/reviewer_papers3} 

\begin{thebibliography}{10}

\bibitem{kddcup}
Kdd cup 2015.
\newblock \url{https://biendata.com/competition/kddcup2015/data/}.
\newblock Accessed: 2018-11-05.

\bibitem{pushshift}
Reddit data dump.
\newblock \url{http://files.pushshift.io/reddit/}.
\newblock Accessed: 2018-11-05.

\bibitem{wikidump}
Wikipedia edit history dump.
\newblock \url{https://meta.wikimedia.org/wiki/Data_dumps}.
\newblock Accessed: 2018-11-05.

\bibitem{agrawal2014big}
D.~Agrawal, C.~Budak, A.~El~Abbadi, T.~Georgiou, and X.~Yan.
\newblock Big data in online social networks: user interaction analysis to
  model user behavior in social networks.
\newblock In {\em International Workshop on Databases in Networked Information
  Systems}, pages 1--16. Springer, 2014.

\bibitem{DBLP:conf/asunam/ArnouxTL17}
T.~Arnoux, L.~Tabourier, and M.~Latapy.
\newblock Combining structural and dynamic information to predict activity in
  link streams.
\newblock In {\em Proceedings of the 2017 {IEEE/ACM} International Conference
  on Advances in Social Networks Analysis and Mining 2017, Sydney, Australia,
  July 31 - August 03, 2017}, pages 935--942, 2017.

\bibitem{DBLP:journals/corr/abs-1804-01465}
T.~Arnoux, L.~Tabourier, and M.~Latapy.
\newblock Predicting interactions between individuals with structural and
  dynamical information.
\newblock {\em CoRR}, abs/1804.01465, 2018.

\bibitem{baytas2017patient}
I.~M. Baytas, C.~Xiao, X.~Zhang, F.~Wang, A.~K. Jain, and J.~Zhou.
\newblock Patient subtyping via time-aware lstm networks.
\newblock In {\em Proceedings of the 23rd ACM SIGKDD International Conference
  on Knowledge Discovery and Data Mining}, pages 65--74. ACM, 2017.

\bibitem{beutel2018latent}
A.~Beutel, P.~Covington, S.~Jain, C.~Xu, J.~Li, V.~Gatto, and E.~H. Chi.
\newblock Latent cross: Making use of context in recurrent recommender systems.
\newblock In {\em Proceedings of the Eleventh ACM International Conference on
  Web Search and Data Mining}, pages 46--54. ACM, 2018.

\bibitem{bobadilla2013recommender}
J.~Bobadilla, F.~Ortega, A.~Hernando, and A.~Guti{\'e}rrez.
\newblock Recommender systems survey.
\newblock {\em Knowledge-based systems}, 46:109--132, 2013.

\bibitem{buntain2014identifying}
C.~Buntain and J.~Golbeck.
\newblock Identifying social roles in reddit using network structure.
\newblock In {\em Proceedings of the 23rd International Conference on World
  Wide Web}, pages 615--620. ACM, 2014.

\bibitem{chaturvedi2014predicting}
S.~Chaturvedi, D.~Goldwasser, and H.~Daum{\'e}~III.
\newblock Predicting instructor's intervention in mooc forums.
\newblock In {\em Proceedings of the 52nd Annual Meeting of the Association for
  Computational Linguistics (Volume 1: Long Papers)}, volume~1, pages
  1501--1511, 2014.

\bibitem{cheng2017anyone}
J.~Cheng, M.~Bernstein, C.~Danescu-Niculescu-Mizil, and J.~Leskovec.
\newblock Anyone can become a troll: Causes of trolling behavior in online
  discussions.
\newblock In {\em CSCW: proceedings of the Conference on Computer-Supported
  Cooperative Work. Conference on Computer-Supported Cooperative Work}, volume
  2017, page 1217. NIH Public Access, 2017.

\bibitem{dai2016deep}
H.~Dai, Y.~Wang, R.~Trivedi, and L.~Song.
\newblock Deep coevolutionary network: Embedding user and item features for
  recommendation.
\newblock {\em arXiv preprint arXiv:1609.03675}, 2016.

\bibitem{DBLP:conf/kdd/DuDTUGS16}
N.~Du, H.~Dai, R.~Trivedi, U.~Upadhyay, M.~Gomez{-}Rodriguez, and L.~Song.
\newblock Recurrent marked temporal point processes: Embedding event history to
  vector.
\newblock In {\em Proceedings of the 22nd {ACM} {SIGKDD} International
  Conference on Knowledge Discovery and Data Mining, San Francisco, CA, USA,
  August 13-17, 2016}, pages 1555--1564, 2016.

\bibitem{DBLP:conf/nips/FarajtabarWGLZS15}
M.~Farajtabar, Y.~Wang, M.~Gomez{-}Rodriguez, S.~Li, H.~Zha, and L.~Song.
\newblock {COEVOLVE:} {A} joint point process model for information diffusion
  and network co-evolution.
\newblock In {\em Advances in Neural Information Processing Systems 28: Annual
  Conference on Neural Information Processing Systems 2015, December 7-12,
  2015, Montreal, Quebec, Canada}, pages 1954--1962, 2015.

\bibitem{ferraz2015rsc}
A.~Ferraz~Costa, Y.~Yamaguchi, A.~Juci Machado~Traina, C.~Traina~Jr, and
  C.~Faloutsos.
\newblock Rsc: Mining and modeling temporal activity in social media.
\newblock In {\em Proceedings of the 21th ACM SIGKDD International Conference
  on Knowledge Discovery and Data Mining}, pages 269--278. ACM, 2015.

\bibitem{DBLP:journals/kbs/GoyalF18}
P.~Goyal and E.~Ferrara.
\newblock Graph embedding techniques, applications, and performance: {A}
  survey.
\newblock {\em Knowl.-Based Syst.}, 151:78--94, 2018.

\bibitem{goyal2018dyngem}
P.~Goyal, N.~Kamra, X.~He, and Y.~Liu.
\newblock Dyngem: Deep embedding method for dynamic graphs.
\newblock {\em arXiv preprint arXiv:1805.11273}, 2018.

\bibitem{grover2016node2vec}
A.~Grover and J.~Leskovec.
\newblock node2vec: Scalable feature learning for networks.
\newblock In {\em Proceedings of the 22nd ACM SIGKDD international conference
  on Knowledge discovery and data mining}, pages 855--864. ACM, 2016.

\bibitem{DBLP:journals/debu/HamiltonYL17}
W.~L. Hamilton, R.~Ying, and J.~Leskovec.
\newblock Representation learning on graphs: Methods and applications.
\newblock {\em {IEEE} Data Eng. Bull.}, 40(3):52--74, 2017.

\bibitem{lastfm}
B.~Hidasi and D.~Tikk.
\newblock Fast als-based tensor factorization for context-aware recommendation
  from implicit feedback.
\newblock In {\em Joint European Conference on Machine Learning and Knowledge
  Discovery in Databases}, pages 67--82. Springer, 2012.

\bibitem{iba2010analyzing}
T.~Iba, K.~Nemoto, B.~Peters, and P.~A. Gloor.
\newblock Analyzing the creative editing behavior of wikipedia editors: Through
  dynamic social network analysis.
\newblock {\em Procedia-Social and Behavioral Sciences}, 2(4):6441--6456, 2010.

\bibitem{johnson2016mimic}
A.~E. Johnson, T.~J. Pollard, L.~Shen, H.~L. Li-wei, M.~Feng, M.~Ghassemi,
  B.~Moody, P.~Szolovits, L.~A. Celi, and R.~G. Mark.
\newblock Mimic-iii, a freely accessible critical care database.
\newblock {\em Scientific data}, 3:160035, 2016.

\bibitem{julier1997new}
S.~J. Julier and J.~K. Uhlmann.
\newblock New extension of the kalman filter to nonlinear systems.
\newblock In {\em Signal processing, sensor fusion, and target recognition VI},
  volume 3068, pages 182--194. International Society for Optics and Photonics,
  1997.

\bibitem{DBLP:journals/corr/abs-1711-10967}
R.~R. Junuthula, M.~Haghdan, K.~S. Xu, and V.~K. Devabhaktuni.
\newblock The block point process model for continuous-time event-based dynamic
  networks.
\newblock {\em CoRR}, abs/1711.10967, 2017.

\bibitem{DBLP:conf/icwsm/Junuthula0D18}
R.~R. Junuthula, K.~S. Xu, and V.~K. Devabhaktuni.
\newblock Leveraging friendship networks for dynamic link prediction in social
  interaction networks.
\newblock In {\em Proceedings of the Twelfth International Conference on Web
  and Social Media, {ICWSM} 2018, Stanford, California, USA, June 25-28,
  2018.}, pages 628--631, 2018.

\bibitem{kloft2014predicting}
M.~Kloft, F.~Stiehler, Z.~Zheng, and N.~Pinkwart.
\newblock Predicting mooc dropout over weeks using machine learning methods.
\newblock In {\em Proceedings of the EMNLP 2014 Workshop on Analysis of Large
  Scale Social Interaction in MOOCs}, pages 60--65, 2014.

\bibitem{kumar2015vews}
S.~Kumar, F.~Spezzano, and V.~Subrahmanian.
\newblock Vews: A wikipedia vandal early warning system.
\newblock In {\em Proceedings of the 21th ACM SIGKDD international conference
  on knowledge discovery and data mining}, pages 607--616. ACM, 2015.

\bibitem{DBLP:conf/cikm/LiDHTCL17}
J.~Li, H.~Dani, X.~Hu, J.~Tang, Y.~Chang, and H.~Liu.
\newblock Attributed network embedding for learning in a dynamic environment.
\newblock In {\em Proceedings of the 2017 {ACM} on Conference on Information
  and Knowledge Management, {CIKM} 2017, Singapore, November 06 - 10, 2017},
  pages 387--396, 2017.

\bibitem{DBLP:journals/access/LiZYZY18}
T.~Li, J.~Zhang, P.~S. Yu, Y.~Zhang, and Y.~Yan.
\newblock Deep dynamic network embedding for link prediction.
\newblock {\em {IEEE} Access}, 6:29219--29230, 2018.

\bibitem{DBLP:conf/sdm/LiDLLGZ14}
X.~Li, N.~Du, H.~Li, K.~Li, J.~Gao, and A.~Zhang.
\newblock A deep learning approach to link prediction in dynamic networks.
\newblock In {\em Proceedings of the 2014 {SIAM} International Conference on
  Data Mining, Philadelphia, Pennsylvania, USA, April 24-26, 2014}, pages
  289--297, 2014.

\bibitem{liyanagunawardena2013moocs}
T.~R. Liyanagunawardena, A.~A. Adams, and S.~A. Williams.
\newblock Moocs: A systematic study of the published literature 2008-2012.
\newblock {\em The International Review of Research in Open and Distributed
  Learning}, 14(3):202--227, 2013.

\bibitem{DBLP:journals/corr/abs-1810-10627}
Y.~Ma, Z.~Guo, Z.~Ren, Y.~E. Zhao, J.~Tang, and D.~Yin.
\newblock Dynamic graph neural networks.
\newblock {\em CoRR}, abs/1810.10627, 2018.

\bibitem{nguyen2018continuous}
G.~H. Nguyen, J.~B. Lee, R.~A. Rossi, N.~K. Ahmed, E.~Koh, and S.~Kim.
\newblock Continuous-time dynamic network embeddings.
\newblock In {\em 3rd International Workshop on Learning Representations for
  Big Networks (WWW BigNet)}, 2018.

\bibitem{DBLP:conf/recsys/PalovicsBKKF14}
R.~P{\'{a}}lovics, A.~A. Bencz{\'{u}}r, L.~Kocsis, T.~Kiss, and E.~Frig{\'{o}}.
\newblock Exploiting temporal influence in online recommendation.
\newblock In {\em Eighth {ACM} Conference on Recommender Systems, RecSys '14,
  Foster City, Silicon Valley, CA, {USA} - October 06 - 10, 2014}, pages
  273--280, 2014.

\bibitem{pennebaker2001linguistic}
J.~W. Pennebaker, M.~E. Francis, and R.~J. Booth.
\newblock Linguistic inquiry and word count: Liwc 2001.
\newblock {\em Mahway: Lawrence Erlbaum Associates}, 71(2001):2001, 2001.

\bibitem{perozzi2014deepwalk}
B.~Perozzi, R.~Al-Rfou, and S.~Skiena.
\newblock Deepwalk: Online learning of social representations.
\newblock In {\em Proceedings of the 20th ACM SIGKDD international conference
  on Knowledge discovery and data mining}, pages 701--710. ACM, 2014.

\bibitem{DBLP:conf/wsdm/QiuDMLWT18}
J.~Qiu, Y.~Dong, H.~Ma, J.~Li, K.~Wang, and J.~Tang.
\newblock Network embedding as matrix factorization: Unifying deepwalk, line,
  pte, and node2vec.
\newblock In {\em Proceedings of the Eleventh {ACM} International Conference on
  Web Search and Data Mining, {WSDM} 2018, Marina Del Rey, CA, USA, February
  5-9, 2018}, pages 459--467, 2018.

\bibitem{raghavan2014modeling}
V.~Raghavan, G.~Ver~Steeg, A.~Galstyan, and A.~G. Tartakovsky.
\newblock Modeling temporal activity patterns in dynamic social networks.
\newblock {\em IEEE Transactions on Computational Social Systems},
  1(1):89--107, 2014.

\bibitem{DBLP:journals/corr/abs-1804-05755}
M.~Rahman, T.~K. Saha, M.~A. Hasan, K.~S. Xu, and C.~K. Reddy.
\newblock Dylink2vec: Effective feature representation for link prediction in
  dynamic networks.
\newblock {\em CoRR}, abs/1804.05755, 2018.

\bibitem{DBLP:journals/corr/abs-1710-00818}
S.~Sajadmanesh, J.~Zhang, and H.~R. Rabiee.
\newblock Continuous-time relationship prediction in dynamic heterogeneous
  information networks.
\newblock {\em CoRR}, abs/1710.00818, 2017.

\bibitem{DBLP:conf/cosn/SedhainSXKTC13}
S.~Sedhain, S.~Sanner, L.~Xie, R.~Kidd, K.~Tran, and P.~Christen.
\newblock Social affinity filtering: recommendation through fine-grained
  analysis of user interactions and activities.
\newblock In {\em Conference on Online Social Networks, COSN'13, Boston, MA,
  USA, October 7-8, 2013}, pages 51--62, 2013.

\bibitem{trivedi2017know}
R.~Trivedi, H.~Dai, Y.~Wang, and L.~Song.
\newblock Know-evolve: Deep temporal reasoning for dynamic knowledge graphs.
\newblock In {\em International Conference on Machine Learning}, pages
  3462--3471, 2017.

\bibitem{trivedi2018representation}
R.~Trivedi, M.~Farajtbar, P.~Biswal, and H.~Zha.
\newblock Representation learning over dynamic graphs.
\newblock {\em arXiv preprint arXiv:1803.04051}, 2018.

\bibitem{walker2015complex}
P.~B. Walker, S.~G. Fooshee, and I.~Davidson.
\newblock Complex interactions in social and event network analysis.
\newblock In {\em International Conference on Social Computing,
  Behavioral-Cultural Modeling, and Prediction}, pages 440--445. Springer,
  2015.

\bibitem{wang2016coevolutionary}
Y.~Wang, N.~Du, R.~Trivedi, and L.~Song.
\newblock Coevolutionary latent feature processes for continuous-time user-item
  interactions.
\newblock In {\em Advances in Neural Information Processing Systems}, pages
  4547--4555, 2016.

\bibitem{wu2017recurrent}
C.-Y. Wu, A.~Ahmed, A.~Beutel, A.~J. Smola, and H.~Jing.
\newblock Recurrent recommender networks.
\newblock In {\em Proceedings of the Tenth ACM International Conference on Web
  Search and Data Mining}, pages 495--503. ACM, 2017.

\bibitem{yang2013turn}
D.~Yang, T.~Sinha, D.~Adamson, and C.~P. Ros{\'e}.
\newblock Turn on, tune in, drop out: Anticipating student dropouts in massive
  open online courses.
\newblock In {\em Proceedings of the 2013 NIPS Data-driven education workshop},
  volume~11, page~14, 2013.

\bibitem{zhang2017deep}
S.~Zhang, L.~Yao, and A.~Sun.
\newblock Deep learning based recommender system: A survey and new
  perspectives.
\newblock {\em arXiv preprint arXiv:1707.07435}, 2017.

\bibitem{zhang2017learning}
Y.~Zhang, Y.~Xiong, X.~Kong, and Y.~Zhu.
\newblock Learning node embeddings in interaction graphs.
\newblock In {\em Proceedings of the 2017 ACM on Conference on Information and
  Knowledge Management}, pages 397--406. ACM, 2017.

\bibitem{zhou2018dynamic}
L.-k. Zhou, Y.~Yang, X.~Ren, F.~Wu, and Y.~Zhuang.
\newblock Dynamic network embedding by modeling triadic closure process.
\newblock In {\em AAAI}, 2018.

\bibitem{zhu2016scalable}
L.~Zhu, D.~Guo, J.~Yin, G.~Ver~Steeg, and A.~Galstyan.
\newblock Scalable temporal latent space inference for link prediction in
  dynamic social networks.
\newblock {\em IEEE Transactions on Knowledge and Data Engineering},
  28(10):2765--2777, 2016.

\bibitem{zhu2017next}
Y.~Zhu, H.~Li, Y.~Liao, B.~Wang, Z.~Guan, H.~Liu, and D.~Cai.
\newblock What to do next: modeling user behaviors by time-lstm.
\newblock In {\em Proceedings of the Twenty-Sixth International Joint
  Conference on Artificial Intelligence, IJCAI-17}, pages 3602--3608, 2017.

\end{thebibliography}

\end{document}